\documentclass[final,oneeqnum, onethmnum, onefignum]{siamltex1213}

\usepackage{amsmath, amsfonts, amssymb, graphicx}

\def\naive{na\"{\i}ve }

\newcommand{\N}{\mathbb N}
\newcommand{\Mat}{\mathbb M}
\newcommand{\A}{\mathbb A}
\newcommand{\Z}{\mathbb Z}
\newcommand{\Zp}{\mathbb{Z}_+}
\renewcommand{\P}{\mathbb P}
\newcommand{\F}{\mathbb F}
\newcommand{\abs}[1]{\lvert #1 \rvert}
\newcommand{\ceil}[1]{\lceil #1 \rceil}

\newcommand{\C}{\mathcal C}
\newcommand{\M}{\mathcal M}
\newcommand{\degree}{\delta} 
\newcommand{\di}{\Delta} 
\newcommand{\Hit}{\mathcal{H}} 

\newcommand{\coeff}{\operatorname{coef}}
\newcommand{\lis}[3]{{#1}_1 #2 {#1}_2 #2\dots #2 {#1}_{#3}}
\newcommand{\comment}[1]{}

\newcommand{\Dx}{D(\mathbf{x})}

\DeclareMathOperator*{\argmin}{arg \, min}
\DeclareMathOperator{\lc}{lc}
\DeclareMathOperator{\Span}{span}
\DeclareMathOperator{\poly}{poly}

\DeclareMathOperator{\Part}{Part}
\DeclareMathOperator{\dist}{d}
\DeclareMathOperator{\nbd}{nbd}
\DeclareMathOperator{\Color}{color}
\DeclareMathOperator{\paths}{paths}
\DeclareMathOperator{\Supp}{S} 
\DeclareMathOperator{\suppo}{s} 
\DeclareMathOperator{\bs}{bs}
\DeclareMathOperator{\bS}{bS}
\DeclareMathOperator{\descend}{descend}

\renewcommand{\sp}{\mathsf{s}} 
\newcommand{\Sp}{\mathsf{S}}
\newcommand{\x}{\mathbf{x}}
\newcommand{\xe}{\mathbf{x}^e}

\newcommand{\xei}{ \mathbf{x}_i^{e_i} }

\renewcommand{\H}{\mathbb H}
\renewcommand{\o}[1]{\overline{#1}}

\newtheorem{construction}[theorem]{Construction}
\newtheorem{observation}[theorem]{Observation}
\newtheorem{claim}[theorem]{Claim}

\newenvironment{reptheorem}[1]{T\textsc{heorem} \ref{#1} (restated).\em}{}
\newenvironment{remark}{R\textsc{emark}.\em}{}
\newenvironment{notation}{N\textsc{otation}.\em}{}

\title{Hitting-sets for ROABP and Sum of Set-Multilinear circuits}

\author{Manindra Agrawal \and Rohit Gurjar \and Arpita Korwar \and Nitin Saxena }

\begin{document}
\maketitle

\begin{abstract}
We give a $n^{O(\log n)}$-time ($n$ is the input size) blackbox 
polynomial identity testing algorithm for unknown-order read-once oblivious
algebraic branching programs (ROABP). The best time-complexity known for this class
was $n^{O(\log^2 n)}$ due to Forbes-Saptharishi-Shpilka (STOC 2014), and
that too only for multilinear ROABP. 
We get rid of their exponential dependence on the individual degree.
With this, we match the time-complexity for the unknown order ROABP with 
the known order ROABP (due to Forbes-Shpilka (FOCS 2013)) and also with the
depth-$3$ set-multilinear circuits (due to Agrawal-Saha-Saxena (STOC 2013)).
Our proof is simpler and involves a new technique called basis isolation.

The depth-$3$ model has recently gained much importance, 
as it has become a stepping-stone to understanding general arithmetic circuits. 
Its restriction to {\em multilinearity} has known exponential lower bounds but 
no nontrivial blackbox identity tests. 
In this paper, we take a step towards designing such hitting-sets. 
We give the first subexponential whitebox PIT 
for the sum of constantly many set-multilinear depth-$3$ circuits.
To achieve this, we define notions of {\em distance} and {\em base sets}. 
Distance, for a multilinear 
depth-$3$ circuit (say, in $n$ variables and $k$ product gates),  
measures how far are the partitions from a mere {\em refinement}. 
The $1$-distance strictly subsumes the set-multilinear model, 
while $n$-distance captures general multilinear depth-$3$. 
We design a hitting-set in time $(nk)^{O(\di \log n)}$ for $\di$-distance. 
Further, we give an extension of our result to models where the distance is large 
(close to $n$) but it is small when restricted to certain base sets (of variables).

We also explore a new model of read-once algebraic branching programs (ROABP) where the factor-matrices are {\em invertible} (called invertible-factor ROABP). We design a hitting-set in time poly($n^{w^2}$) for width-$w$ invertible-factor ROABP. Further, we could do {\em without} the invertibility restriction when $w=2$. Previously, the best result for width-$2$ ROABP was quasi-polynomial time (Forbes-Saptharishi-Shpilka, STOC 2014).

\end{abstract}


\section{Introduction}
The problem of {\em Polynomial Identity Testing} is that of
deciding if a given polynomial is nonzero. The complexity of 
the question depends crucially on the way the polynomial is input to the PIT test.
For example, if the polynomial is given as a set of coefficients
of the monomials, then we can easily check whether the polynomial is nonzero
in polynomial time. 
The problem has been studied for different input models.
Most prominent among them is the model of arithmetic circuits. 
Arithmetic circuits are the arithmetic analog of boolean circuits
and are defined over a field $\F$.
They are directed acyclic graphs, where every node
is a `$+$' or `$\times$' gate and each input gate is 
a constant from the field $\F$ or a variable from
$\mathbf{x} = \{\lis{x}{,}{n}\}$. 
Every edge has a weight from the underlying field $\F$.
The computation is done in the natural way.
Clearly, the output gate computes a polynomial
in $\F[\o{x}]$. 
We can restate the PIT problem as: Given an arithmetic circuit $\C$, 
decide if the polynomial computed by $\C$ is nonzero in time polynomial
in the circuit size. 
Note that, given a circuit, computing the polynomial explicitly
is not possible, as it can have exponentially many monomials. 
However, given the circuit, 
it is easy to compute an evaluation of the polynomial
by substituting the variables with constants. 

Though there is no known {\em deterministic} algorithm for PIT,
there are easy randomized algorithms, e.g.\ \cite{Sch80}.
These randomized algorithms are based on the theorem:
A nonzero polynomial,
evaluated at a random point, gives a nonzero value with a good probability. 
Observe that such an algorithm does not need to access 
the structure of the circuit, it just uses the evaluations;
it is a {\em blackbox} algorithm.
The other kind of algorithms, where the structure of the
input is used, are called {\em whitebox} algorithms. 
Whitebox algorithms for PIT have many known applications. E.g.\ graph matching reduces to PIT.
On the other hand, blackbox algorithms (or \emph{hitting-sets}) have connections to circuit lower bound proofs. Arguably, this is currently the only concrete approach towards lower bounds, see \cite{mul12, Mul12b}. See the surveys by Saxena \cite{Sax09, Sax14} and Shpilka \& Yehudayoff \cite{SY10} for more
applications.

An Arithmetic Branching Program (ABP) 
is another interesting model of computing polynomials. 
It consists of a directed acyclic graph with a source and a sink. 
The edges of the graph have polynomials as their weights.
The weight of a path is the product of the weights of the edges
present in the path. 
The polynomial computed by the ABP
is the sum of the weights of all the paths from the source to the sink.
It is well known that for an ABP, the underlying graph 
can seen as a layered graph such that all paths from the source to
the sink have exactly one edge in each layer. 
And the polynomial computed by the ABP can be written as 
a \emph{matrix product}, where each matrix corresponds to a layer. 
The entries in the matrices are weights of the corresponding edges. 
The maximum number of vertices in a layer, 
or equivalently, the dimension of the corresponding matrices 
is called the \emph{width }of the ABP.
It is known that symbolic determinant and ABP are equivalent models of computation
\cite{Tod91, MV97}.
Ben-Or \& Cleve \cite{BOC92} have shown that 
a polynomial computed by a formula of logarithmic depth and constant fan-in,
can also be computed by a width-$3$ ABP.
Thus, ABP is a strong model for computing polynomials.
The following chain of reductions shows the power of ABP and its constant-width 
version relative to
other arithmetic computation models (see \cite{BOC92} and \cite[Lemma 1]{Nis91}).
\begin{eqnarray*} 
&\text{Constant-depth Arithmetic Circuits} 
\leq_p  \text{Constant-width ABP} &\\
 & \leq_p  \text{Formulas} 
 \leq_p  \text{ABP}
 \leq_p  \text{Arithmetic Circuits}&
\end{eqnarray*}

Our first result is for a special class of ABP called 
{\em Read Once Oblivious Arithmetic Branching Programs (ROABP)}.
An ABP is a read once ABP (ROABP) if the weights in its $n$ layers
are univariate polynomials in $n$ distinct variables, i.e.\ 
the $i$-th layer has weights coming from $\F[x_{\pi(i)}]$, 
where $\pi$ is a permutation on the set $\{1, 2, \dots, n\}$.
When we know this permutation $\pi$, we call it an ROABP with {\em known} variable order
(it is significant only in the blackbox setting).

Raz and Shpilka \cite{RS05} gave a $\poly(n, w, \degree)$-time whitebox
algorithm for $n$-variate polynomials computed by a width-$w$ ROABP with
individual degree bound $\degree$.
Recently, Forbes and Shpilka \cite{FS12,FS13} gave a $\poly(n, w, \degree)^{\log n}$-time blackbox 
algorithm for the same, when the variable order is known.
Subsequently, Forbes et al.\ \cite{FSS13} gave a 
blackbox test for the case of unknown variable order, but 
with time complexity being $\poly(n)^{\degree \log w \log n}$.
Note the exponential dependence on the degree. 
Their time complexity becomes quasi-polynomial in case of multilinear polynomials, 
i.e.\ $\degree =1$.

In another work Jansen et al. \cite{JQS10a} gave quasi-polynomial time blackbox test for a sum of
constantly many multilinear ``ROABP". Their definition of ``ROABP" is more stringent. They assume
that every variable appears in at most once in the ABP. 
Later, this result was generalized to ``read-$r$ OABP" \cite{JQS10b}, where a variable can occur 
in at most one layer, and on at most $r$ edges.
Our definition of ROABP seems much more powerful than both of these.

We improve the result of \cite{FSS13} and match the time complexity for the unknown
order case with the known order case (given by \cite{FS12,FS13}).
Unlike \cite{FSS13}, we do not have exponential dependence on the individual degree.
Formally,

\begin{theorem}
\label{thm:ROABPhs}
Let $C(\mathbf{x})$ be an $n$-variate polynomial computed by a
width-$w$ ROABP (unknown order) with the degree of each variable bounded by $\degree$. 
Then there is a $\poly(n,w,\degree)^{\log n}$-time hitting set for $C$.
\end{theorem}

\begin{remark}
Our algorithm also works when the layers have their weights as general sparse 
polynomials (still over disjoint sets of variables) instead of 
univariate polynomials (see the detailed version in Section~\ref{sec:ROABP}).
\end{remark}

A polynomial computed by a width-$w$ ABP can be written as $S^{\top} D(\x) T$,
 where $S,T \in \F^w$ and $D(\x) \in \F^{w \times w}[\x]$ is a polynomial over 
the matrix algebra.
Like \cite{ASS13,FSS13}, we try to construct a basis (or extract the rank)
 for the coefficient vectors in $D(\x)$.
We actually construct a weight assignment on the variables, which {\em isolates} a
basis in the coefficients in $D(\x)$. 
This idea is inspired from the rank extractor techniques in \cite{ASS13,FSS13}.
Our approach is to directly work with $D(\x)$, while \cite{ASS13,FSS13}
have applied a rank extractor to small subcircuits of $D(\x)$, by shifting it carefully.
In fact, the idea of {\em basis isolating weight assignment} evolved when
we tried to find a direct proof, for the rank extractor in \cite{ASS13}, which
does not involve subcircuits.
But, our techniques go much further than both \cite{ASS13,FSS13}, as is evident
from our strictly better time-complexity results.

The boolean analog of ROABP, read once ordered branching programs (ROBP)
have been studied extensively, with regard to the RL vs. L question.
For ROBP, a pseudorandom generator (PRG) with seed length $O(\log^2 n)$ 
($n^{O(\log n)}$ size sample set)
is known in the case of known variable order \cite{N90}.
This is analogous to the \cite{FS13} result for known order ROABP.
On the other hand, in the unknown order case, the best known seed length
is of size $n^{1/2 + o(1)})$ ($2^{{n}^{1/2+o(1)}}$ size sample set) \cite{IMZ12}.
One can ask: Can the result for the unknown order case be matched with
the known order case in the boolean setting as well.
Recently, there has been a partial progress in this direction by \cite{SVW14}.

The PIT problem has also been studied for various restricted classes of circuits. 
One such class is depth-$3$ circuits. 
Our second result is about a special case of this class.
A depth-$3$ circuit is usually
defined as a $\Sigma \Pi \Sigma$ circuit: The circuit gates are in 
three layers, the top layer has an output gate which is $+$, second layer has
all $\times$ gates and the last layer has all $+$ gates. 
In other words, the polynomial computed by a $\Sigma\Pi\Sigma$ circuit is of the form 
$C(\overline{x}) = \sum_{i=1}^k a_i \prod_{j=1}^{n_i}  \ell_{ij}$, 
where $n_i$ is the number of input lines to the $i$-{th} product gate 
and $\ell_{ij}$ is a linear polynomial of the form $b_0 + \sum_{r=1}^n b_r x_r$.
An efficient solution for depth-$3$ PIT is still not known.
Recently, it was shown by Gupta et al.\ \cite{GKKS13}, that depth-3 circuits are almost as powerful as general circuits.
A polynomial time hitting-set for a depth-$3$ circuit implies a quasi-poly-time hitting-set for general circuits.
Till now, for depth-$3$ circuits, efficient PIT is known
when the top fan-in is assumed to be constant \cite{DS07, KS07, KS09, KS11, SS11, SS12, SS13}
and for certain other restrictions \cite{Sax08, SSS13, ASSS12}.

On the other hand, there are exponential lower bounds for depth-$3$ 
{\em multilinear} circuits \cite{RY09}. 
Since there is a connection between lower bounds and PIT \cite{Agr05}, 
we can hope that solving PIT for depth-$3$ multilinear circuits 
should also be feasible. 
This should also lead to new tools for general depth-$3$. 

A polynomial is said to be multilinear if 
the degree of every variable in every term is at most $1$.
The circuit $C(\o{x})$ is a multilinear circuit 
if the polynomial computed at every gate is multilinear.
A polynomial time algorithm is known only for a sub-class of multilinear 
depth-$3$ circuits, called {\em depth-$3$ set-multilinear circuits}.
This algorithm is due to Raz and Shpilka \cite{RS05} and is whitebox.
In a depth-$3$ multilinear circuit, 
since every product gate computes a multilinear polynomial, 
a variable occurs in at most one of the $n_i$ 
linear polynomials input to it. 
Thus, each product gate naturally induces a 
{\em partition} of the variables, where
each {\em color} (i.e.~part) of the partition contains the variables
present in a linear polynomial $\ell_{ij}$.  
Further, if the partitions induced by all the 
$k$ product gates are the same then the circuit is
called a depth-$3$ set-multilinear circuit.

Agrawal et al.\ \cite{ASS13} gave a quasi-polynomial time blackbox 
algorithm for the class of depth-$3$ set-multilinear circuits.  
But till now, no subexponential time test (not even whitebox) was known even for sum of two 
set-multilinear circuits.
We give a subexponential time whitebox PIT for sum of constantly many
set-multilinear circuits.

\begin{theorem}
\label{thm:c-setmultihs}
Let $C(\mathbf{x})$ be a $n$-variate polynomial, which is a sum of $c$ 
set-multinear depth-$3$ circuits, each having top fan-in $k$. 
Then there is a $n^{O(2^{c-1} n^{1 - \epsilon} \log k) }$-time whitebox test for $C$,
where $\epsilon := 1/2^{c-1}$.
\end{theorem}

To achieve this, we define a new class of circuits, as a tool, called
{\em multilinear depth-$3$ circuits
with $\di$-distance}. 
A multilinear depth-$3$ circuit has $\di$-distance
if there is an ordering 
on the partitions induced by the product gates, 
say $(\lis{\P}{,}{k})$, such that for any color 
in the partition $\P_i$, there exists a set of $\le (\di-1)$ 
other colors in $\P_i$ such that the set of variables in 
the union of these $\le \di$ colors are {\em exactly} partitioned
in the upper partitions, i.e.\ $\{\lis{\P}{,}{i-1}\}$.
As we will see,
such sets of $\di$ colors form equivalence classes of the colors at partition $\P_i$.
We call them friendly neighborhoods and
they help us in identifying subcircuits.
Intuitively, the distance measures how far away are the partitions 
from a mere \emph{refinement} sequence of partitions, $\lis{\P}{\le}{k}$.
A refinement sequence of partitions will have distance $1$.
On the other hand, general multilinear depth-$3$ circuits can have at most
$n$-distance.

As it turns out, a polynomial computed by a depth-$3$ $\di$-distance
circuit (top fan-in $k$) can also be computed by a width-$O(kn^{\di})$ ROABP
(see Lemma~\ref{lem:dDistROABP}).
Thus, we get a $\poly(nk)^{\di \log n}$-time hitting set for this class,
 from Theorem~\ref{thm:ROABPhs}.
Next, we use a general result about finding a hitting set for a class
{\em $m$-base-sets-$\mathsf{C}$}, 
if a hitting set is known for class $\mathsf{C}$.
A polynomial is in $m$-base-sets-$\mathsf{C}$,
if there exists a partition of the variables into $m$ base sets
such that restricted to each base set (treat other variables as field constants), 
the polynomial is in class $\mathsf{C}$.
We combine these two tools to prove Theorem~\ref{thm:c-setmultihs}.
We show that a sum of constantly many set-multilinear circuits 
falls into the class $m$-base-sets-$\di$-distance, for $m \di = o(n)$.

Agrawal et al.\ \cite{AGKS13} had achieved {\em rank concentration}, which implies
a hitting set, for the class
$m$-base-sets-$\di$-distance, but through complicated proofs.
On the other hand, this work gives only a hitting set
for the same class, but with the advantage of simplied proofs.

Our third result deals again with arithmetic branching programs. 
The results of \cite{BOC92} and \cite{SSS09} show
that the constant-width ABP is already a strong model. 
Here, we study constant-width ABP with some natural restrictions.

We consider a class of ROABPs where all the matrices in the matrix product, except
the left-most and the right-most matrices, are invertible. 
We give a blackbox test for this class of ROABP. 
In contrast to \cite{FSS13} and 
our Theorem~\ref{thm:ROABPhs}, this test works in \emph{polynomial time}
if the dimension of the matrices is constant.

Note that the class of ABP, where the factor matrices are invertible, is quite powerful,
as Ben-Or and Cleve \cite{BOC92} actually reduce formulas to width-$3$ ABP 
with \emph{invertible} factors. 
Saha, Saptharishi and Saxena \cite{SSS09} reduce depth-$3$ circuits to width-$2$ ABP with invertible factors. 
But the constraints of invertibility and read-once together seem
to restrict the computing power of ABP.
Interestingly, an analogous class of read-once boolean branching programs called
{\em permutation branching programs} 
has been studied recently \cite{KNP11, De11, Ste12}.
These works give
PRG for this class (for constant width) with seed-length $O(\log n)$,
in the known variable order case.
In other words, they give polynomial size sample set which can fool these programs.
For the unknown variable order case,
 Reingold et al.\ \cite{RSV13} gave a PRG
with seed-length $O(\log^2 n)$.
Our polynomial size hitting sets for the arithmetic setting work for 
any unknown variable order.
Hence, it is better as compared to the currently known results for the boolean case.

\begin{theorem}[Informal version]
\label{thm:invROABPHS}
Let $C(\o{x}) = D_0^{\top} (\prod_{i=1}^{d} D_i ) D_{d+1}$
be a polynomial such that $D_0 \in \F^w[{x}_{j_0}]$ and $D_{d+1} \in \F^w[{x}_{j_{d+1}}]$
and for all $i \in [d]$, $D_i \in \F^{w \times w}[x_{j_i}]$ is an invertible matrix (order of the variables is unknown).
Let the degree bound on $D_i$ be $\delta$ for $0 \leq i \leq d+1$.
Then there is a $\poly((\delta n)^{w^2})$-time hitting-set for $C(\o{x})$. 
\end{theorem}

The proof technique here is very different from the first two theorems 
(here we show {\em rank concentration} over a {\em non-commutative} algebra, 
see the proof idea in Section~\ref{sec:invROABP}).
Our algorithm works even when the factor matrices have their entries as general sparse 
polynomials (still over disjoint sets of variables) instead of 
univariate polynomials (see the detailed version in Section~\ref{sec:invROABP}).
Running time in this case grows to quasi-polynomial (but is still better
than Theorem~\ref{thm:ROABPhs} in several interesting cases).

If the matrices are $2 \times 2$, then we do not need the assumption
of invertibility (see Theorem~\ref{thm:ROABP22HS}, Section~\ref{sec:2ROABP}).
So, for width-$2$ ROABP our results are strictly stronger than \cite{FSS13}
and our Theorem~\ref{thm:ROABPhs}.
Here again, there is a comparable result in the boolean setting. 
PRG with seed-length $O(\log n)$ (polynomial size sample set) are known 
for width-$2$ ROBP \cite{BDVY13}.

\section{Preliminaries}
\label{sec:preliminaries}
\paragraph{Hitting Set} A set of points $\Hit$ is called a hitting set
for a class $\mathsf{C}$ of polynomials if for any nonzero polynomial $P$
in $\mathsf{C}$, there exists a point in $\Hit$ where $P$ evaluates to a nonzero
value. 
An $f(n)$-time hitting set would mean that the hitting set can be generated 
in time $f(n)$ for input size $n$.

\subsection{Notation}
$\Zp$ denotes the set $\N \cup \{0\}$.
$[n]$ denotes the set $\{1,2, \dots, n\}$. 
$[[n]]$ denotes the set $\{0,1, \dots, n\}$. 
$\mathbf{x}$ will denote a set of variables.
For a set of $n$ variables $\mathbf{x} = \{x_1, x_2, \dots, x_n\}$
and for an exponent $\mathbf{e} = (e_1, e_2, \dots, e_n) \in \Zp^n$,
$\xe$ will denote the monomial $\prod_{i=1}^n x_i^{e_i}$.
The \emph{support} of a monomial is the set of 
variables that have degree $ \ge 1 $ in that monomial. 
The \emph{support size }of the monomial is the cardinality of its support. 
A polynomial is called $s$-{\em sparse} if there are $s$ monomials in it
with nonzero coefficients.
For a polynomial $P$, the coefficient of the monomial $m$ in $P(\x)$ 
is denoted by $\coeff_P(m)$.

$\F^{m \times n}$ represents the set of all $m \times n$ matrices over the field $\F$.
$\Mat_{m \times m}(\F)$ will denote the algebra of $m \times m$ matrices over the field $\F$. 
Let $\A_k(\F)$ be any $k$-dimensional algebra over the field $\F$.
For any two elements $A = (a_1, a_2, \dots a_k) \in \A_k(\F)$
and $B = (b_1, b_2, \dots b_k) \in \A_k(\F)$
(having a natural basis representation in mind), 
their dot product
is defined as $A \cdot B = \sum_{i = 1}^n a_k b_k$;
and the product $AB$ will denote the product in the algebra $\A_k(\F)$.

$\Part(S)$ denotes the set of all possible partitions of the set $S$. 
Elements in a partition are called {\em colors} (or parts).

\subsection{Arithmetic Branching Programs}
An ABP is a directed graph with $d+1$ layers of vertices
$\{V_0,V_1, \dots, V_{d}\}$ and a start node $u$ 
and an end node $t$
such that the edges are 
only going from $u$ to $V_0$, $V_{i-1}$ to $V_i$ for any $i \in [d]$,
$V_d$ to $t$. 
A width-$w$ ABP has $\abs{V_i} \leq w$ for all $i \in [[d]]$.
Let the set of nodes in $V_i$ be $\{v_{i,j} \mid j \in [w]\}$.
All the edges in the graph have weights from $\F[\mathbf{x}]$,
for some field $\F$. As a convention, the edges going from $s$
and coming to $t$ are assumed to have weights from the field $\F$.

For an edge $e$, let us denote
its weight by $W(e)$.
For a path $p$ from $u$ to $t$,
its weight $W(p)$ is defined to be the product of weights of all the edges
in it, i.e.\ $\prod_{e \in p} W(e)$. 
Consider the polynomial $C(\mathbf{x}) = \sum_{p \in \paths(u,t)} W(p)$ 
which is the sum of the weights of all the paths from $u$ to $t$.
This polynomial $C(\mathbf{x})$ is said to be computed by the ABP.

It is easy to see that this polynomial is the same as
$S^{\top} (\prod_{i=1}^{d} D_i ) T $,
where $S,T \in \F^w$ 
and $D_i$ is a $w \times w$ matrix for $1 \leq i \leq d$ such that 
\begin{eqnarray*}
S(\ell) &=& W(u,v_{0,\ell}) \text{ for } 1 \leq \ell \leq w\\
D_i(k, \ell) &=& W(v_{i-1,k},v_{i,\ell}) \text{ for } 1 \leq \ell,k \leq w \text{ and } 1 \leq i \leq d\\
T(k) &=& W(v_{d,k},t) \text{ for } 1 \leq k \leq w
\end{eqnarray*}

\subsubsection*{ROABP} An ABP is called a {\em read once oblivious ABP (ROABP)}
if the edge weights in the different layers 
are univariate polynomials in distinct variables.
Formally, the entries in $D_{i}$ come from $\F[x_{\pi(i)}]$
for all $i \in [d]$, where $\pi$ is a permutation on the set $[d]$.

\subsubsection*{sparse-factor ROABP} 
We call the ABP a {\em sparse-factor ROABP} if the edge weights in different
layers are sparse polynomials in disjoint sets of variables.
Formally, 
if there exists an unknown partition of the variable set $\mathbf{x}$
into $d$ sets $\{\lis{\mathbf{x}}{,}{d}\}$
such that
$D_{i} \in \F^{w \times w}[{\mathbf{x}}_{i}]$ is a $s$-sparse polynomial,
for all $i \in [d]$,
then the corresponding ROABP is called a {\em $s$-sparse-factor} ROABP.
It is read once in the sense that in the corresponding 
ABP, any particular variable contributes to at most one edge 
on any path.

\subsection{Kronecker Map}
We will often use a weight function on the variables which separates 
a desired set of monomials. 
Let $w \colon {\mathbf x} \to \N$ be a weight function on the variables. 
Consider its natural extension to the set of all monomials 
$w \colon \Zp^n \to \N$ as follows: 
$w(\Pi_{i=1}^n x_i^{\gamma_i}) = \sum_{i=1}^n \gamma_i w(x_i)$, 
where $\gamma_i \in \Zp , \; \forall i \in [n]$. 

\begin{lemma}[Efficient Kronecker map \cite{Kro1882,Agr05}]
\label{lem:kronecker}
Let $\M$ be the set of all monomials in $n$ variables
$\mathbf{x} = \{x_1, x_2, \dots, x_n\}$ with maximum individual
degree $\degree$. Let $A$ be a set of pairs of monomials from 
$\M$.  
Then there exists a (constructible) set of $N$-many weight functions
$w \colon \mathbf{x} \to [1, \dots, N \log N]$, such that at 
least one of them separates all the pairs in $A$, i.e.\
for any $(m,m') \in A$, $w(m) \neq w(m')$, 
where $N := O(n \abs{A} \log(\degree+1))$.

\end{lemma}
\begin{proof}
Since we want to separate the $n$-variate monomials with maximum 
individual degree $\degree$, 
we use the \naive Kronecker map $W \colon x_i \mapsto (\degree +1)^{i-1}$ 
for all $i \in [n]$. 
It can be easily seen that $W$ will give distinct weights to
any two monomials (with maximum individual degree $\degree$).
But, the weights given by $W$ are exponentially high. 

So, we take the weight function $W$ modulo $p$, for many small primes $p$. 
 Each prime $p$ leads to a different weight function. 
That is our set of candidate weight functions. 
We need to bound the number $N$ of primes 
that ensures that at least one of the weight functions 
separates all the monomial pairs in $A$. 
We choose the smallest $N$ primes, say $\mathcal{P}$ is the set. 
By the effective version of the Prime Number Theorem, 
the highest value in the set $\mathcal{P}$ is $N \log N$.

To bound the number $N$ of primes: We want a $p \in \mathcal{P}$ such that 
$\forall (m,m') \in A, \; W(m) - W(m') \not\equiv 0 \pmod{p}$.
 Which means, 
$$ \exists p \in \mathcal{P}, \;	
p \nmid \prod_{(m,m') \in A} \left( W(m) - W(m') \right).$$ 
In other words,
$$\prod_{p \in \mathcal{P}} p \nmid \prod_{(m,m') \in A} \left(W(m) - W(m')\right).$$ 
This can be ensured by setting 
$\prod_{p \in \mathcal{P}} p > \prod_{(m,m') \in A} \left(W(m) - W(m')\right)$.
There are $\abs{A}$ such monomial pairs and 
each $W(m) < n \degree (\degree+1)^{n -1}$. 
Also, $\prod_{p \in \mathcal{P}} p > 2^N$. Hence, $N = O(n \abs{A} \log(\degree+1))$ suffices.
\end{proof}


\section{Hitting set for ROABP: Theorem~\ref{thm:ROABPhs}}
\label{sec:ROABP}
Like \cite{ASS13} and \cite{FSS13}, we work with the vector polynomial.
I.e.\ for a polynomial computed by a width-$w$ ROABP, 
$C(\mathbf{x}) = S^{\top} (\prod_{i=1}^d D_i) T$, we see the product
$D := \prod_{i=1}^d D_i$ as a polynomial over the matrix algebra $\Mat_{w \times w}(\F)$.
We can write the polynomial $C(\mathbf{x})$ as the dot product
$R \cdot D$, where $R = ST^{\top}$.
The vector space spanned by the coefficients of $D(\mathbf{x})$ is called
the coefficient space of $D(\mathbf{x})$.
This space will have dimension at most $w^2$. 
We essentially try to construct a small set of vectors, by evaluating $D(\mathbf{x})$,
 which can span
the coefficient space of $D(\mathbf{x})$.  
Clearly, if $C \neq 0$ then the dot product of $R$ with at least one of these 
spanning vectors will be nonzero. 
And thus, we get a hitting set.

Unlike \cite{ASS13} and \cite{FSS13}, we directly work with the original 
polynomial $D(\x)$, instead of shifting it and breaking it into subcircuits. 
Our approach for finding the hitting set is to come up with a weight function
on the variables which can {\em isolate a basis}
for the coefficients of the polynomial $D(\x)$. 
This can be seen as a generalization of isolating a monomial for a polynomial in 
$\F[\mathbf{x}]$, which is a usual technique for PIT (e.g.\ sparse PIT \cite{KS01}).

We present our results for polynomials over arbitrary algebra.
Let $\A_k(\F)$ be a $k$-dimensional algebra over the field $\F$. 
Let $\mathbf{x} = \{x_1, x_2, \dots, x_n\}$ be a set of variables and
let $D(\mathbf{x})$ be a polynomial in $\A_k(\F)[\x]$ with highest individual 
degree $\degree$. 
Let $\M$ denote the set of all monomials over the variable set ${\mathbf x}$
with highest individual degree $\degree$.

Now, we will define a {basis isolating weight assignment} for a polynomial
$D \in \A_k(\F)[\mathbf{x}]$ which would lead to a hitting set for the polynomial
$C \in \F[\mathbf{x}]$, where $C = R \cdot D$, for some $R \in \A_k(\F)$.


\begin{definition}[Basis Isolating Weight Assignment]
A weight function $w \colon \mathbf{x} \to \N$ is called a basis isolating
weight assignment for a polynomial $D(\mathbf{x}) \in \A_k(\F)[\mathbf{x}]$ 
if there exists 
a set of monomials $S \subseteq \M$ ($k' := \abs{S} \leq k$)
whose coefficients form a basis for the coefficient space of $D({\mathbf x})$, 
such that 
\begin{itemize}
\item for any $m, m' \in S$, $w(m) \neq w(m')$ and
\item for any monomial $m \in \M \setminus S$, 
$$ \coeff_D(m) \in \Span \{ \coeff_D(m') \mid m' \in S, \; w(m') < w(m) \}.$$
\end{itemize}
\end{definition}

The above definition is equivalent to saying that there exists
a {\em unique minimum} weight basis (according to the weight function $w$)
among the coefficients of $D$, and also the basis monomials have distinct weights. 
We skip the easy proof for this equivalence, as we will not need it. 
Note that a weight assignment, which gives distinct weights to all the monomials,
is indeed a basis isolating weight assignment. 
But, it will involve exponentially large weights. To, find an efficient
weight assignment one must use some properties of the given circuit.
First, we show how such a weight assignment would lead to hitting set.
We will actually show that it isolates a monomial in $C(\x)$.

\begin{lemma}
\label{lem:basisIsolation}
Let $w \colon \mathbf{x} \to \N$ is a basis isolating weight assignment
for a polynomial $D(\mathbf{x}) \in \A_k(\F)[\mathbf{x}]$. 
And let $C = R \cdot D$ be a nonzero polynomial, for some $R \in \A_k(\F)$.
Then, after the substitution $x_i = t^{w(x_i)} $ for all $i \in [n]$,
the polynomial $C$ remains nonzero, where $t$ is an indeterminate.
\end{lemma}
\begin{proof}
Let $D_m \in \A_k(\F)$ denote the coefficient $\coeff_D(m)$.
It is easy to see that after the mentioned substitution, 
the new polynomial $C'(t)$ 
is equal to
$\sum_{m \in \M} (R \cdot D_m) t^{w(m)}$. 

Let us say that $S \subset \M$ is the set of monomials whose coefficients
form the isolated basis for $D$.
According to the definition of the basis isolating weight assignment,
for any monomial $m \in \M \setminus S$, 
\begin{equation}
\label{eq:span}
 D_m \in \Span \{D_{m'} \mid m' \in S, \; w(m') < w(m) \}.
\end{equation}

First, we claim that $\exists m' \in S$ such that $ R \cdot D_{m'} \neq 0$.
For the sake of contradiction, let us assume that 
$\forall m' \in S, \; R \cdot D_{m'} = 0$.
Taking the dot product with $R$ on both the sides of Equation~(\ref{eq:span}),
we get that for any monomial $m \in \M \setminus S$,
$$ R \cdot D_m \in \Span \{R \cdot D_{m'} \mid m' \in S, \; w(m') < w(m) \}.$$
Hence, $R \cdot D_m = 0, \; \forall m \in \M$. 
That means $C(\x) =0$, which contradicts our assumption.

Now, let $m^*$ be the minimum weight monomial in $S$ whose coefficient
gives a nonzero dot product with $R$, i.e.\ 
$m^* = \displaystyle\argmin_{m \in S} \{w(m) \mid R \cdot D_m \neq 0\}$.
There is a unique such monomial in $S$ because all the monomials in 
$S$ have distinct weights.

We claim that $\coeff_{C'}(t^{w(m^*)}) \neq 0$
and hence $C'(t) \neq 0$.
To see this, consider any monomial $m$, other than $m^*$, with $w(m) = w(m^*)$.
The monomial $m$ has to be in the set $\M \setminus S$, as the monomials in $S$ 
have distinct weights.
From Equation~(\ref{eq:span}),
$$
 D_m \in \Span \{D_{m'} \mid m' \in S, \; w(m') < w(m^*) \}.
$$
Taking dot product with $R$ on both the sides we get,
$$
 R \cdot D_m \in \Span \{ R \cdot D_{m'} \mid m' \in S, \; w(m') < w(m^*) \}.
$$
But, by the choice of $m^*$, $R \cdot D_{m'} = 0$,
for any $m' \in S$ with $w(m') < w(m^*)$.
Hence, $R \cdot D_m = 0$, for any $m \neq m^*$ with $w(m) = w(m^*)$.

So, the coefficient $\coeff_{C'}(t^{w(m^*)})$ can be written as
$$ \sum_{\substack {m \in \M \\ w(m) = w(m^*)} } R \cdot D_m = 
R \cdot D_{m^*},
 $$
which, we know, is nonzero.
\end{proof}

To construct a hitting set for $C$, we can try many possible field values
of $t$. The number of such values needed will be the degree of $C$
after the substitution, which is at most $(n \degree \max_{i} w(x_i))$.
Hence, the cost of the hitting set is dominated by the 
{\em cost of the weight function}, i.e.\ the maximum weight given to any variable
and the time taken to construct the weight function.

In the next step, we show that such a basis isolating weight assignment
can indeed be found for a sparse-factor ROABP,
but with cost quasi-polynomial in the input size. 
First, we make the following observation that
it suffices that the coefficients of the monomials not in $S$,
linearly depend on any coefficients with strictly smaller weight, not
necessarily coming from $S$. 

\begin{observation}
\label{obs:smallerWeight}
If, for a polynomial $D \in \A_k(\F)[\x]$, there exists 
a weight function $w \colon \mathbf{x} \to \N$ 
and a set of monomials $S \subseteq \M$ ($k' := \abs{S} \leq k$)
such that 
for any monomial $m \in \M \setminus S$, 
$$ \coeff_D(m) \in \Span \{ \coeff_D(m') \mid m' \in \M , \; w(m') < w({m}) \}.$$
then
we can also conclude that
for any monomial $m \in \M \setminus S$, 
$$ \coeff_D(m) \in \Span \{ \coeff_D(m') \mid m' \in S, \; w(m') < w({m}) \}.$$
\end{observation}
\begin{proof}
We are given that for any monomial $m \in \o{S} := \M \setminus S$, 
$$ \coeff_D(m) \in \Span \{ \coeff_D(m') \mid m' \in \M, \; w(m') < w({m}) \}. $$
Any coefficient $\coeff_D({m'})$ on the right hand side of this equation, 
which corresponds to an index in $\o{S}$, 
can be replaced with some other coefficients, which have further
smaller weight.   
If we keep doing this, we will be left 
with the coefficients only corresponding to the set $S$, 
because in each step we are getting smaller and smaller weight coefficients. 
\end{proof}

In our construction of the weight function, 
we will create the set $\o{S} := \M \setminus S$ incrementally, i.e.\
in each step we will
make more coefficients depend on strictly smaller weight
coefficients.
Finally, we will be left with only $k'$ 
(the rank of the coefficient space of $D$) many coefficients in $S$.
We present the result for an arbitrary $k$-dimensional algebra $\A_k(\F)$,
 instead of just the matrix algebra.

\begin{lemma}[Weight Construction]
\label{lem:weightFunction}
Let $\mathbf{x}$ be given by a union of $d$ disjoint sets of variables
${\mathbf x}_1 \sqcup {\mathbf x}_2 \sqcup \dotsm \sqcup {\mathbf x}_d$, 
with $\abs{\mathbf{x}} = n$. 
Let $D({\mathbf x}) = P_1({\mathbf x}_1) P_2({\mathbf x}_2) \dotsm P_d({\mathbf x}_d)$,
where $P_i \in \A_k(\F)[{\mathbf x}_i]$ is a sparsity-$s$, 
individual degree-$\degree$ polynomial, for all $i \in [d]$. 
Then, we can construct a basis isolating weight assignment 
for $D(\mathbf{x})$ 
 with the cost being $(\poly(k, s, n, \degree))^{\log d}$.
\end{lemma}
\begin{proof}
In our construction, 
the final weight function $w$ will be a combination of $(\log d +1)$-many
different weight functions, say $(w_0, w_1, \dots, w_{\log d})$. 
Let us say,
 their precedence is decreasing from left to right, 
i.e.\ $w_0$ has the highest
precedence and $w_{\log d}$ has the lowest precedence. 
As mentioned earlier, we will build the set $\o{S}$ 
(the set of monomials whose coefficients are 
in the span of strictly smaller weight coefficients than themselves)
 incrementally
in $(\log d +1)$ steps, using weight function $w_i$ in the $(i+1)$-th step. 

{\em Iteration $0$}: Let $\M_{0,1}, \M_{0,2}, \dots, \M_{0,d}$ be 
the sets of monomials and
 $\C_{0,1}, \C_{0,2}, \dots, \C_{0,d}$ be the sets of coefficients in the polynomials 
$P_1, P_2, \dots, P_d$ respectively. 

\begin{notation}
The product of two sets of monomials $\M_1$ and $\M_2$ is defined as
$\M_1 \times \M_2 = \{ m_1 m_2 \mid m_1 \in \C_1, \; m_2 \in \C_2 \} $.
The product of any two sets of coefficients $\C_1$ and $\C_2$ is defined as
$\C_1 \times \C_2 = \{ c_1 c_2 \mid c_1 \in \C_1, \; c_2 \in \C_2 \} $.
\end{notation}

The crucial property of the polynomial $D$ is that the set of coeffcients 
in $D$, $\C_0$, is just the product $\C_{0,1} \times \C_{0,2} \times \dotsm \times \C_{0,d}$.
Similary, the set of all the monomials in $D$, say $\M_0$, can
be viewed as the product $\M_{0,1} \times \M_{0,2} \times \dotsm \times \M_{0,d}$.
Let $m := m_{a} m_{a+1} \dotsm m_{b}$ be a monomial, where $1 \leq a \leq b \leq d$
and $m_{j} \in \M_{0,j}$, for $a \leq j \leq b$.
Then $D_m$ will denote the coefficient 
$ \coeff_{P_a}(m_a) \coeff_{P_{a+1}}(m_{a+1}) \dotsm \coeff_{P_b}(m_b)$.


Let us fix $w_0 \colon \x \to \N$ to be a weight function on the variables
which gives distinct weights to
 all the $s$ monomials in $\M_{0,i}$, for each $i \in [d]$.
As $w_0$ assigns distinct weights to these monomials, 
so does the weight function $w$.

For each $P_i$ we do the following: 
\begin{itemize}
\item arrange the coefficients in $\C_{0,i}$ 
in increasing order of their weight according to $w$ 
(or equivalently, according to $w_0$),
\item choose a maximal set of linearly independent coefficients, in a greedy manner, 
going from lower weights to higher weights. 
\end{itemize}
The fact that the weight functions $w_1, w_2, \dots, w_{\log d}$
are not defined yet does not matter because
$w_0$ has the highest precedence. 
The total order given to the monomials in $\M_{0,i}$ by $w_0$
is the same as given by $w$, irrespective of what the functions 
$w_1, \dots, w_{\log d}$ are chosen to be.

This gives us a basis for the
coefficients of $P_i$, say $\C'_{0,i}$. 
Let $\M'_{0,i}$ denote the monomials in $P_i$ corresponding to these
basis coefficients. 
From the construction of the basis, it follows that
for any monomial $m \in \M_{0,i} \setminus \M'_{0,i}$ ,
\begin{equation}
\label{eq:dependsOnLower}
D_m \in \Span \{ D_{m'} \mid m' \in \M'_{0,i}, \; w(m') < w(m) \} .
\end{equation}

Now, consider any monomial $m \in \M$ which is not present in   
the set $\M'_0 := \M'_{0,1} \times \M'_{0,2} \times \dotsm \times \M'_{0,d}$. 
Let $m = m_1 m_2 \dotsm m_d$, where $m_i \in \M_{0,i}$ for all $i \in [d]$.
We know that for at least one $j \in [d]$, $m_j \in \M_{0,j} \setminus \M'_{0,j}$.
Then using Equation~(\ref{eq:dependsOnLower}) we can write the following about
$D_m = D_{m_1} D_{m_2} \dotsm D_{m_d}$,
$$ D_m\in \Span
\{ D_{m_1} \dotsm D_{m_{j-1}} D_{m'_j} D_{m_{j+1}} \dotsm D_{m_d}
\mid m'_j \in \M'_{0,j}, \; w({m'_j}) < w({m_j})
 \}.
 $$
This holds, because the algebra product is bilinear. 
Equivalently, for any monomial $m \in \M_0 \setminus \M'_0$,
$$ D_m \in \Span \{ D_{m'} \mid m' \in \M_0, \; w(m') < w({m}) \} . $$
This is true because 
$$w({m_1}) + \dotsm + w({m'_j}) + \dotsm + w({m_d}) 
< w({m_1}) + \dotsm + w({m_j}) + \dotsm + w({m_d}) =
w({m})
.$$

Hence, all the monomials in $\M_0 \setminus \M'_0$ can be put
into $\o{S}$, i.e.\ their corresponding coefficients depend
on strictly smaller weight coefficients.

{\em Iteration $1$}: Now, let us consider monomials in the set 
$\M'_0 = \M'_{0,1} \times \M'_{0,2} \times \dotsm  \times \M'_{0,d}$.
Let the corresponding set of coefficients be 
$\C'_0 := \C'_{0,1} \times \C'_{0,2} \times \dotsm \times \C'_{0,d}$.
Since, the underlying algebra $\A_k(\F)$ has dimension at most $k$
and the coefficients
in $\C'_{0,i}$ form a basis for $\C_{0,i}$,
$\abs{\M'_{0,i}} \leq k$, for all $i \in [d]$. 
In the above product, let us make $d/2$ disjoint pairs of consecutive terms,
and for each pair, multiply the two terms in it. 
Putting it formally, 
let us define $\C_{1,j}$ to be the product $\C'_{0,2j-1} \times \C'_{0,2j}$
and similarly $\M_{1,j} := \M'_{0,2j-1} \times \M'_{0,2j}$, for all $j \in [d/2]$ 
(if $d$ is odd, we can make it even by multiplying the identity element of
$\A_k(\F)$ in the end). 
Now, let $\C_1 := \C'_0 = \C_{1,1} \times \C_{1,2} \times \dotsm \times \C_{1,d_1}$, 
and $\M_1 := \M'_0 = \M_{1,1} \times \M_{1,2} \times \dotsm  \times \M_{1,d_1}$,
where $d_1 := d/2$. 
For any $i \in [d_1]$, $\M_{1,i}$ has at most $k^2$ monomials. 

Now, we fix the weight function $w_1 \colon \mathbf{x} \to \N$ such that
it gives distinct weights to all the monomials in $\M_{1,i}$, for each $i \in [d_1]$. 
As $w_1$ separates these monomials, so does the weight function $w$.
Now, we repeat the same procedure of constructing a 
basis in a greedy manner for $\C_{1,i}$ according to the weight
function $w$, for each $i \in [d_1]$.
Let the basis coefficients for $\C_{1,i}$
be $\C'_{1,i}$ and corresponding monomials be $\M'_{1,i}$. 

As argued before, any coefficient in $\C_1$, 
which is outside the set 
$\C'_1 := \C'_{1,1} \times \C'_{1,2} \times \dotsm \times \C'_{1,d_1}$,
is in the span of strictly smaller weight (than itself) coefficients. 
So, we can also put the corresponding monomials 
$\M_1 \setminus \M'_1$ in $\o{S}$
where $\M'_1 := \M'_{1,1} \times \M'_{1,2} \times \dotsm  \times \M'_{1,d_1}$.

{\em Iteration $r$}: We keep repeating the same procedure for $(\log d +1)$-many rounds.
After round $r$, say the set of monomials we are left with is given by
the product
$\M'_{r-1} = \M'_{r-1,1} \times \M'_{r-1,2} \times \dotsm \times \M'_{r-1,d_{r-1}}$,
where $\M_{r-1,i}$ has at most $k$ monomials, for each $i \in [d_{r-1}]$ 
and $d_{r-1} = d / 2^{r-1}$.
In the above product, 
we make $d_{r-1}/2$ disjoint pairs of consecutive terms, 
and multiply the two terms in each pair.
Let us say we get $\M_r := \M'_{r-1}= 
\M_{r,1} \times \M_{r,2} \times \dotsm \times \M_{r,d_r}$,
where $d_r = d_{r-1}/2$.
Say, the corresponding set of coefficients is given by
$\C_r = \C_{r,1} \times \C_{r,2} \times \dotsm \times \C_{r,d_r}$.
Note that $\abs{\M_{r,i}} \leq k^2$, for each $i \in [d_r]$.

We fix the weight function $w_r$ such that it gives distinct weights to
all the monomials in the set $\M_{r,i}$, for each $i \in [d_r]$.
We once again mention that fixing of $w_r$ does not affect 
the greedy basis constructed in earlier rounds 
and hence the monomials which were put in the set $\o{S}$,
because $w_r$ has less precendence than any $w_{r'}$, for $r' < r$.

For each $\C_{r,i}$, we construct a basis in a greedy manner 
going from lower weight to higher weight 
(according to the weight function $w$).
Let this set of basis coefficients be $\C'_{r,i}$ and corresponding
monomials be $\M'_{r,i}$, for each $i \in [d_r]$.
Let $\C'_{r}  := \C'_{r,1} \times \C'_{r,2} \times \dotsm \times \C'_{r,d_r}$
and $\M'_{r} := \M'_{r,1} \times \M'_{r,2} \times \dotsm \times \M'_{r,d_r}$.
Arguing similar as before we can say that
 each coefficient in $\C_{r,i} \setminus \C'_{r,i}$ 
is in the span of strictly smaller weight coefficients (from $\C'_{r,i}$) than itself.
Hence, the same can be said about any coefficient in the set  
$\C_r \setminus \C'_r$. 
So, all the monomials in the set
$\M_r \setminus \M'_r$
can be put into $\o{S}$.
Now, we are left with monomials 
$\M'_r = \M'_{r,1} \times \M'_{r,2} \times \dotsm \times \M'_{r,d_r}$
for the next round. 

{\em Iteration $\log d$}: As in each round, the number of terms in the product gets halved,
after $\log d$ rounds we will be left with just one term,
i.e.\ $\M_{\log d} = \M'_{\log d -1, 1} \M'_{\log d -1 ,2} = \M_{\log d ,1}$. 
Now, we will fix the function $w_{\log d}$ which separates
all the monomials in $\M_{\log d, 1}$.
By arguments similar as above, we will be finally left with
at most $k'$ monomials in $S$, which will all have distinct weights. 
It is clear that for every monomial in $\o{S}$, its coefficient will be in 
the span of strictly smaller weight coefficients than itself.

Now, let us look at the cost of this weight function. 
In the first round, $w_0$ needs to separate at most 
$O(d s^2)$ many pairs of monomials. 
For each $1 \leq r \leq \log d$, $w_r$ 
needs to separate at most $O(dk^4)$ many pairs of monomials. 
From Lemma~\ref{lem:kronecker}, to construct $w_r$, for any $0 \leq r \leq \log d$,
one needs to try $\poly(k, s, n, \degree)$-many weight functions
each having highest weight at most $\poly(k, s, n, \degree)$ 
(as $d$ is bounded by $n$). 
To get the correct combination of the weight functions $(w_0, w_1, \dots, w_{\log d})$
we need to try all possible combinations of these polynomially many choices 
for each $w_r$.
Thus, we have to try $(\poly(k, s, n, \degree))^{\log d}$ many combinations. 

To combine these weight functions we can choose a large enough number $B$ 
(greater than the highest weight a monomial can get in any of the weight functions),
and define $w := w_0 B^{\log d} + w_1 B^{\log d -1} + \dotsm + w_{\log d}$. 
The choice of $B$ ensures that the different weight functions cannot interfere with each 
other, and they also get the desired precedence order. 

The highest weight a monomial can get from the weight function $w$ 
would be $(\poly(k, s, n, \degree))^{\log d}$. Thus, the cost of $w$
remains $(\poly(k, s, n, \degree))^{\log d}$.

\end{proof}

Combining Lemma~\ref{lem:weightFunction} with 
 Observation~\ref{obs:smallerWeight} and Lemma~\ref{lem:basisIsolation}, 
we can get a hitting set for ROABP.

\begin{reptheorem}{thm:ROABPhs}
Let $C(\x)$ be an $n$-variate polynomial computed by a width-$w$, $s$-sparse-factor
ROABP, with individual degree bound $\degree$. Then there is a $\poly(w, s, n, \degree)^{\log n}$-time
hitting set for $C(\x)$.
\end{reptheorem}
\begin{proof}
As mentioned earlier, $C(\x)$ can be written as $R \cdot D(\x)$,
for some $R \in \Mat_{w \times w}(\F)$, where $D(\x) \in \Mat_{w \times w}(\F)[\x]$.
The underlying matrix algebra $\Mat_{w \times w}(\F)$ has dimension $w^2$.
The hitting set size will be dominated by the cost of 
the weight function constructed in Lemma~\ref{lem:weightFunction}.
As the parameter $d$ in Lemma~\ref{lem:weightFunction}, i.e.\ the number of layers in
the ROABP, is bounded by $n$,
the hitting set size will be $\poly(w, s, n, \degree)^{\log n}$.
\end{proof}

\section{Sum of constantly many set-multilinear circuits: Theorem~\ref{thm:c-setmultihs}}
\label{sec:sumOfSetMulti}
To find a hitting set for a sum of constantly many set-multilinear circuits, we 
build some tools. The first is depth-3 multilinear circuits with `small distance'.
As it turns out, a multilinear polynomial computed by a depth-$3$ $\di$-distance
circuit (top fan-in $k$) can also be computed by a width-$O(kn^{\di})$ ROABP (Lemma~\ref{lem:dDistROABP}).
Thus, we get a $\poly(nk)^{\di \log n}$-time hitting set for this class,
 from Theorem~\ref{thm:ROABPhs}.
Next, we use a general result about finding a hitting set for a class
{\em $m$-base-sets-$\mathsf{C}$}, 
if a hitting set is known for class $\mathsf{C}$ (Lemma~\ref{lem:baseSets}).
A polynomial is in $m$-base-sets-$\mathsf{C}$,
if there exists a partition of the variables into $m$ base sets
such that restricted to each base set (treat other variables as field constants), 
the polynomial is in class $\mathsf{C}$.
Finally, we show that a sum of constantly many set-multilinear circuits 
falls into the class $m$-base-sets-$\di$-distance, for $m \di = o(n)$.
Thus, we get Theorem~\ref{thm:c-setmultihs}.

\subsection{$\di$-distance circuits}\label{subsec:deltaDistance}

Recall that each product gate in a depth-$3$ multilinear circuit 
induces a partition on the variables. Let these partitions be $ \lis{\P}{,}{k} $.

\begin{definition}[Distance for a partition sequence]
\label{def:distance}
Let $\P_1, \P_2, \dots, \P_k \in \Part ([n])$ be the $k$ partitions of the variables $\{x_1, x_2, \dots, x_n\}$. Then $\dist(\P_1, \P_2, \dots, \P_k) = \di$ if 
 $ \forall i \in \{2,3,\dots , k\}, \forall\text{colors } Y_1 \in \P_i, \; \exists Y_2, Y_3, \dots, Y_{\di'} \in \P_i \; (\di' \le \di)$ such that $\lis{Y}{\cup}{\di'}$ equals a union of some colors in $ \P_j, \forall j\in [i-1]$.
\end{definition}

In other words, in every partition $\P_i$,
each color $Y_1$ has a set of colors called `friendly neighborhood',
$\{\lis{Y}{,}{\di'}\}$,
consisting of at most $\di$ colors,
which is exactly partitioned in the `upper partitions'.
We call $\P_i$, an {\em upper} partition relative to $\P_j$ 
(and $\P_j$, a {\em lower} partition relative to $\P_i$),
if $i < j$.
For a color $X_a$ of a partition $\P_j$, let $\nbd_j (X_a)$ denote its friendly neighborhood.
The friendly neighborhood $\nbd_j (x_i)$ of a variable $x_i$ in a partition $\P_j$ is defined as $\nbd_j (\Color_j(x_i))$, where $\Color_j(x_i)$ is the color in the partition $\P_j$ that contains the variable $x_i$.



\begin{definition}[$ \di $-distance circuits]
A multilinear depth-$3$ circuit $ C $
has $ \di $-distance if its product gates can be ordered to correspond to a partition sequence $(\P_1,\dots,\P_k)$ with $\dist (\lis{\P}{,}{k}) \le \di$.

\end{definition}


Every depth-$3$ multilinear circuit is thus an $n$-distance circuit. 
A circuit with a partition sequence, 
where the partition $\P_i$ is a refinement of the partition $\P_{i+1}, 
\forall i \in [k-1]$, exactly characterizes a $1$-distance circuit. 
All depth-$3$ multilinear circuits have distance between $1$ and $n$.
Also observe that the circuits with $1$-distance strictly subsume set-multilinear circuits.
E.g.\ a circuit, whose product gates
induce two different partitions
$\P_1 = \{\{1\}, \{2\}, \dots, \{n\}\}$ and 
$\P_2 = \{ \{1,2\} , \{3,4\} , \dots, \{n-1,n\} \}$, 
has $1$-distance but is not set-multilinear.

\noindent \textbf{Friendly neighborhoods -} To get a better picture, we ask: Given a color $X_a$ of a partition $\P_j$ in a circuit $D(\mathbf{x})$,
how do we find its friendly neighborhood $\nbd_j(X_a)$?
Consider a graph $G_j$ which has the colors of the partitions $\{ \P_1, \P_2, \dots, \P_j \}$, 
as its vertices. 
For all $i \in [j-1]$, there is an edge between the colors $X \in \P_{i}$ and 
$Y \in \P_{j}$ if they share at least one variable.
Observe that if any two colors $X_a$ and $X_b$ of partition $\P_j$
are reachable from each other in $G_j$,
then, they should be in the same neighborhood.
As reachability is an equivalence relation, {\em the neighborhoods are equivalence classes of colors}. 

Moreover, observe that for any two variables $x_a$ and $x_b$,
if their respective colors in partition $\P_j$, $\Color_j(x_a)$ and $\Color_j(x_b)$
are reachable from each other in $G_j$
then their respective colors in partition $\P_{j+1}$, 
$\Color_{j+1}(x_a)$ and $\Color_{j+1}(x_b)$
are also reachable from each other in $G_{j+1}$. Hence,
\begin{observation}
\label{obs:zeroBelow}
If at some partition,
the variables $x_a$ and $x_b$ are in the same neighborhood,
then, they will be in the same neighborhood in all of the lower partitions.
I.e.\ $\nbd_j(x_a) = \nbd_j(x_b) \implies \nbd_i(x_a) = \nbd_i(x_b), \forall i \ge j$.
\end{observation}
In other words, if we define a new sequence of partitions,
such that the $j$-th partition has $x_a$ and $x_b$ in the same color if $\nbd_j(x_a) = \nbd_j(x_b)$, 
then the upper partitions are {\em refinements} of the lower partitions. 


%
\subsubsection{Reduction to ROABP}
Now, we show that any polynomial computed by a low-distance multilinear 
depth-$3$ circuit can also be computed by a small size ROABP. 
First we make the following observation about sparse polynomials. 

\begin{observation}
\label{obs:sparse}
Any multilinear polynomial $C(\mathbf{x})$ with sparsity $s$
can be computed by a width-$s$ ROABP, in any variable order. 
\end{observation}
\begin{proof}
Let $\M$ denote the set of monomials in $C$, and 
let $C_m$ denote $\coeff_C(m)$.
Consider an ABP with 
$n+1$ layers of vertices $V_1, V_2, \dots, V_{n+1}$ each having
$s$ vertices (one for each monomial in $\M$)
 together with a start vertex $v_0$ and an end 
vertex $v_{n+2}$.
Let $v_{i,m}$ denote the $m$-th vertex of the layer $V_i$,
for any $i \in [n+1]$ and any $m \in \M$.

The edge labels in the ABP are given as follows:
For all $m \in \M$,
\begin{itemize} 
\item The edge $(v_0, v_{1,m})$ is labelled by $C_m$,
\item The edge $(v_{n+1, m}, v_{n+2})$ is labelled by $1$,
\item For all $i \in [n]$, 
the edge $(v_{i,m}, v_{i+1,m})$ is labelled by $x_i$ if the monomial
$m$ contains $x_i$, otherwise by $1$.
\end{itemize}
All other edges get labelled by $0$.
Clearly, the ABP constructed computes the polynomial $P(\mathbf{x})$ 
and it is an ROABP.

Also, note that this construction can be done with any desired variable order.
\end{proof}

Now, consider a depth-$3$ $\di$-distance multilinear polynomial 
$P = \sum_{i = 1}^k a_i Q_i$, where each $Q_i = \prod_{j=1}^{n_i} \ell_{ij}$ 
is a product of linear polynomials. 
We will construct an ROABP for each $Q_i$. 
We can combine these ROABPs to construct a single ROABP 
if they all have the same variable order. 
To achieve this we use the {\em refinement} property 
described above (from Observation~\ref{obs:zeroBelow}).

\begin{lemma}
Let $P = \sum_{i = 1}^k a_i Q_i$ be a polynomial computed
by a $\di$-distance circuit. Then we can make a width-$O(n^{\di})$ ROABP 
for each $Q_i$, in the same variable order.
\end{lemma}
\begin{proof}
Each $Q_i$ is a product of linear forms in disjoint set of variables, 
say $Q_i = \prod_{j=1}^{n_i} \ell_{ij}$.
Let the partition induced on the variable set, by the product $Q_i$,
be $\P_i$, for all $i \in [k]$.
Without loss of generality let the partition sequence 
$(\P_1, \P_2, \dots, \P_k)$ have distance $\di$.
For each $i \in [k]$, let us define a new partition 
$\P'_i$, such that the union of colors in each neighborhood of $\P_i$
forms a color of $\P'_i$.
This is a valid definition, as neighborhoods are equivalence classes of
colors. 
From Observation~\ref{obs:zeroBelow}, the partition $\P'_i$
is a refinement of partition $\P'_j$ for any $i<j$.

For a partition $\P$ of the variable set $\mathbf{x}$, 
an ordering on its colors
$(c_1 < c_2 < \dotsm < c_r)$ naturally induces a partial ordering 
on the variables, i.e.\ for any $x_i \in c_{j}$ and $x_{i'} \in c_{j'}$,
$c_j < c_{j'} \implies x_i < x_{i'}$.
The variables in the same color do not have any relation.

Let us say, a variable (partial) order
$(<^*)$ {\em respects} 
a partition $\P$ with colors $\{c_1, c_2, \dots, c_r\}$,
if there exists an ordering of the colors $(c_{j_1} < c_{j_2} < \dots < c_{j_r})$,
such that its induced partial order $(<)$ on the variables can be extended to $<^*$.
We claim that there exists a variable order $(<^*)$
which {\em respects} partition $\P'_i$, for all $i \in [k]$.

We build this variable order $(<^*)$ iteratively.
We start with $\P'_k$. 
We give an arbitrary ordering to the colors in $\P'_k$, 
say $(c_{k,1} < c_{k,2} < \dots < c_{k,r_k})$, which induces
a partial order $(<_k)$ on the variables. 
For any $k > i \geq 1$, let us define a partial order $(<_i)$
inductively as follows:
Let $(<_{i+1})$ be a partial order on the variables induced
by an ordering on the colors of $\P'_{i+1}$.
As mentioned earlier, the colors of $\P'_i$ are just further partitions
of the colors of $\P'_{i+1}$. 
Hence, we can construct an ordering on the colors of $\P'_i$,
such that the induced partial order $(<_i)$ is an extension of $(<_{i+1})$.
To achieve that, we do the following: For each color $c$ in $\P'_{i+1}$,
fix an arbitrary ordering among those colors of $\P'_i$, whose union forms $c$.

Clearly, the partial order $(<_1)$ defined in such a way respects $\P'_i$
for all $i \in [k]$.
We further fix an arbitrary ordering among variables belonging to the same color
in $\P'_1$.
Thus, we get a total order $(<^*)$, which is an extension of $<_1$ and hence respects
$\P'_i$ for all $i \in [k]$.

Now, we construct an ROABP for each $Q_i$ in the variable order $<^*$.
First, we multiply out the linear forms which belong to the same neighborhood
in each $Q_i$. 
That is, we write $Q_i$ as the product $\prod_{j=1}^{r_i} Q_{ij}$,
where $r_i$ is the number of neighborhoods in $\P_i$ 
(number of colors in $\P'_i$) and
each $Q_{ij}$ is the product of linear forms (colors) 
which belong to the same neighborhood in $\P_i$. 
As, the partition sequence has distance $\di$, the neighborhoods
have at most $\di$ colors.
So, the degree of each $Q_{ij}$ is bounded by $\di$ and hence
the sparsity is bounded by $O(n^{\di})$.
By Observation~\ref{obs:sparse}, we can construct a width-$O(n^{\di})$ 
ROABP for $Q_{ij}$ in the variable order given by $<^*$.

Let $c_{ij}$ denote the color of $\P'_i$ corresponding to $Q_{ij}$.
As the order $<^*$ respects $\P'_i$, it gives an order on its colors,
say $c_{ij_1} < c_{ij_2} < \dots < c_{ij_{r_i}}$.
Now, we arrange the ROABPs for $Q_{ij}$'s in the order
$Q_{ij_1} Q_{ij_2} \dots Q_{ij_{r_i}}$, while identifying
the end vertex of $Q_{i j_{a}}$ with the start vertex of
$Q_{i j_{a+1}}$, for all $a \in [r_i-1]$.
Clearly the ROABP thus constructed computes the polynomial
$Q_i$ and has variable order $<^*$.

\end{proof}
 
Once we have ROABPs for the polynomials $Q_i$'s in the same
variable order, let us make a new start node and connect it with 
the start node of the ROABP for $Q_i$ with label $a_i$,
for all $i \in [k]$. Also, let us make a new end node and connect
it with the end node of the ROABP for $Q_i$ with label $1$,
for all $i \in [k]$. Clearly, the ROABP thus constructed computes
the polynomial $P = \sum_{i = 1}^k a_i Q_i$ and has width $O(kn^{\di})$.
Thus, we can write

\begin{lemma}[$\di$-distance to ROABP]
\label{lem:dDistROABP}
An $n$-variate polynomial computed by a depth-$3$, $\di$-distance circuit 
with top fan-in $k$ has a width-$O(kn^{\di})$
ROABP.
\end{lemma}

Hence, from Theorem~\ref{thm:ROABPhs} we get,
\begin{theorem}[$\di$-distance Hitting Set]
\label{thm:dDistancehs}
Let $C(\mathbf{x})$ be a depth-$3$, $\di$-distance, 
$n$-variate multilinear circuit with top fan-in $k$. Then
there is a $(nk)^{O(\di \log n)}$-time hitting-set for $C(\mathbf{x})$. 
\end{theorem}

\subsection{Base sets with $\di$-distance}
\label{sec:baseSets}
In this section we describe our second tool towards finding
a hitting set for sum of constantly many set-multilinear polynomials.
We further generalize the class 
of polynomials, for which we can give an efficient test, 
beyond low-distance. 
Basically, it is enough to have low-distance ``projections". 

\begin{definition}
A multilinear depth-$3$ circuit $C(\mathbf{x})$ is said to have 
$m$-base-sets-$\di$-distance
if there is a partition of the 
variable set $\mathbf{x}$
into base sets $\{\lis{\mathbf{x}}{,}{m}\}$ such that 
for any $i \in [m]$, 
restriction of $C$ on the $i$-{th} base set 
(i.e.\ other variables are considered as field constants), 
 has $\di$-distance. 
\end{definition}

We will show that there is an efficient hitting set for this class of polynomials.
In fact, we can show a general easy result for a polynomial whose
restriction on one base set falls into a class $\mathsf{C}$,
for which a hitting set is already known. 

\begin{lemma}[Hybrid Argument]
\label{lem:baseSets}
Let $\Hit$ be the hitting set for a class of (n-variate) polynomials $\mathsf{C}$.
Let $\mathbf{x}$ be a union of $m$ disjoint sets of variables 
$\mathbf{x}_1 \sqcup \mathbf{x}_2 \sqcup \dotsm \sqcup \mathbf{x}_m $, 
called base sets, each with size at most $n$.
Let $C(\mathbf{x})$ be a polynomial such that its restriction to the
base set $\mathbf{x}_i$ (i.e.\ the other variables are considered as field constants), 
is in class $\mathsf{C}$, for all $i \in [m]$.
Then there is a hitting set for $C(\mathbf{x})$ of size ${\abs{\Hit}}^m$ 
(with the knowledge of the base sets).
\end{lemma}
\begin{proof}
Let us assume that the set $\mathbf{x}_i$ has cardinality $n$, for all $i \in [m]$. 
If not, then we can introduce dummy variables. 
Now, we claim that if $C(\mathbf{x}) \neq 0$ then there exists $m$ points 
$\mathbf{h}_1, \mathbf{h}_2, \dots, \mathbf{h}_m \in \Hit$, such that 
$C(\mathbf{x}_1 =  \mathbf{h}_1, \mathbf{x}_2 =  \mathbf{h}_2, \mathbf{x}_m =  \mathbf{h}_m) \neq 0$.

We prove the claim inductively.

\emph{Base Case:} The polynomial $C(\mathbf{x}_1, \mathbf{x}_2, \dots, \mathbf{x}_m) \neq 0$.
It follows from the assumption.

\emph{Induction Hypothesis:} There exists points $\mathbf{h}_1, \mathbf{h}_2, \dots, \mathbf{h}_i \in \Hit$ 
such that the partially evaluated polynomial 
$C'(\mathbf{x_{i+1}}, \dots, \mathbf{x}_m) := 
C(\mathbf{x}_1 =  \mathbf{h}_1, \dots, \mathbf{x}_i =  \mathbf{h}_i, \mathbf{x}_{i+1}, \dots, \mathbf{x}_m) \neq 0$.

\emph{Induction Step:} We show that there exists $\mathbf{h_{i+1}} \in \Hit$ such that 
the polynomial $C'(\mathbf{x}_{i+1}=\mathbf{h}_{i+1}, \mathbf{x}_{i+2}, \dots, \mathbf{x}_m) \neq 0$.

The polynomial $C'$ is nothing but the polynomial 
$C$ evaluated at $\mathbf{x}_1, \dots, \mathbf{x}_i$.
Hence, the polynomial $C'$ restricted to 
the set $\mathbf{x}_{i+1}$, is also in the class $\mathsf{C}$.
So, there must exist a point $\mathbf{h}_{i+1} \in \Hit$ 
such that $C'(\mathbf{x}_{i+1} = \mathbf{h}_{i+1}) \neq 0$.

Thus, the claim is true.
Now, to construct a hitting set for $C$, one needs to substitute the set $\Hit$
for each base set $\mathbf{x}_i$, i.e.\ the cartesian product 
$\Hit \times \Hit \times \dots \times \Hit$ ($m$ times).
Hence, we get a hitting set of size ${\abs{\Hit}}^m$.

\end{proof}

Note that, in the above proof the knowledge of the base sets is crucial.
This lemma, together with Theorem~\ref{thm:dDistancehs}, gives us the following:

\begin{theorem}[$m$-base-sets-$\di$-distance PIT]
\label{thm:baseSetsHS}
If $C(\mathbf{x})$ is a depth-$3$ multilinear circuit, 
with top fan-in $k$, 
having $m$ base sets (known) with $\di$-distance, 
then there is a $(nk)^{O(m \di \log n)}$-time hitting-set for $C$. 
\end{theorem}

\subsection{Sum of set-multilinear circuits reduces to $m$-base-sets-$\di$-distance}
\label{sec:sumSetMult}
In this section, we will reduce the PIT for 
sum of constantly many set-multilinear depth-$3$ circuits, 
to the PIT for depth-$3$ circuits with 
$m$-base-sets-$\di$-distance, where $m \di = o(n)$. 
Thus, we get a subexponential time whitebox algorithm 
for this class (from Theorem~\ref{thm:baseSetsHS}). 
Note that a sum of constantly many set-multilinear depth-$3$ circuits
is equivalent to a depth-$3$ multilinear circuit such that 
the number of distinct partitions, induced by its product gates,
is constant.

We first look at the case of two partitions. 
For a partition $\P$ of $[n]$, 
let $\P|_B$ denote the restriction of $\P$ on a base set
$B \subseteq [n]$. 
E.g., if $\P = \{ \{1,2\}, \{3,4\}, \{5,6, \dots, n\}\}$ and $B = \{1,3,4\}$
then $\P|_B = \{ \{1\}, \{3,4\} \}$.
Recall that $d(\P_1,\P_2, \dots, \P_c)$ denotes the {\em distance}
of the partition sequence $(\P_1,\P_2, \dots, \P_c)$ (Definition~\ref{def:distance}).  
For a partition sequence $(\P_1, \P_2, \dots \P_c)$, 
and a base set $B \subseteq [n]$, 
let $d_B(\P_1, \P_2, \dots, \P_c)$ denote the distance of the partition sequence
when restricted to the base set $B$, i.e.\
$d(\P_1|_B, \P_2|_B, \dots, \P_c|_B)$.

\begin{lemma}
\label{lem:twoPartitions}
For any two partitions $\{ \P_1, \P_2 \}$ of the set $[n] $,
there exists a partition of $[n]$,
into at most $2 \sqrt{n}$ base sets $\{B_1, B_2, \dots, B_{m} \}$
 $(m < 2 \sqrt{n})$,
such that for any $i \in [m]$,
either $d_{B_i}(\P_1, \P_2) = 1$ or $d_{B_i}(\P_2, \P_1) =1$.
\end{lemma}          
\begin{proof}
Let us divide the set of colors in the partition $\P_1$, into 
two types of colors: One with at least $\sqrt{n}$ elements and the other with less than
$\sqrt{n}$ elements. 
In other words, $\P_1 = \{X_1, X_2, \dots, X_{r} \} \cup \{ Y_1, Y_2, \dots, Y_{q}\}$
such that $\abs{X_i} \geq \sqrt{n}$ and $\abs{Y_j} < \sqrt{n}$,
 for all $i \in [r], \; j \in [q]$.
Let us make each $X_i$ a base set, i.e.\ $B_i = X_i$, $\forall i \in [r]$. 
As $\abs{X_i} \geq \sqrt{n}, \; \forall i \in [r]$, we get $r \leq \sqrt{n}$.
Now, for any $i \in [r]$, $\P_1|_{B_i}$ has only one color.
Hence, irrespective of what colors $\P_2|_{B_i}$ has, 
$d_{B_i}(\P_2, \P_1) = 1$, for all $i \in [r]$.

Now, for the other kind of colors, we will make base sets  
which have exactly one element from each color $Y_j$.
More formally, let $Y_j = \{y_{j,1}, y_{j,2}, \dots, y_{j, r_j} \}$, for all $j \in [q]$.
Let $r' = \max\{ r_1, r_2, \dots, r_q\}$ ($r' < \sqrt{n}$). 
Now define base sets $B'_1, B'_2, \dots, B'_{r'}$ such that
for any $a \in [r']$, $B'_a = \{y_{j,a} \mid j \in [q], \; \abs{Y_j} \geq a \} $. 
In other words, all those $Y_j$s which have at least $a$ elements, contribute
their $a$-th element to $B'_a$.
Now for any $a \in [r']$, $\P_1|_{B'_a} = \{ \{y_{j,a}\} \mid j \in [q], \; \abs{Y_j} \geq a \}$,
i.e.\ it has exactly one element in each color. 
Clearly, irrespective of what colors $\P_2|_{B'_a}$ has, 
$d_{B'_a} (\P_1,\P_2) = 1$, for all $a \in [r']$.

$\{B_1, B_2, \dots, B_r\} \cup \{B'_1, B'_2, \dots, B'_{r'}\}$ is our final set of base sets. 
Clearly,  they form a partition of $[n]$. 
The total number of base sets, $m = r + r' < 2 \sqrt{n}$.

\end{proof}

Now, we generalize Lemma~\ref{lem:twoPartitions} to any constant number of partitions, by induction. 

\begin{lemma}[Reduction to $m$-base-sets-$1$-distance]
\label{lem:cPartitions}
For any set of $c$ partitions $\{ \P_1, \P_2, \dots, \P_c \} \subseteq \Part([n])$,
there exists a partition of the set $[n]$,
into $m$ base sets $\{B_1, B_2, \dots, B_{m} \}$
with $m < 2^{c-1} \cdot n^{1 - (1/2^{c-1})}$
such that for any $i \in [m]$,
there exists a permutation of the partitions, 
$(\P_{i_1} , \P_{i_2}, \dots, \P_{i_c})$
with $d_{B_i} (\P_{i_1} , \P_{i_2}, \dots, \P_{i_c}) = 1$.
\end{lemma}
\begin{proof}
Let $f(c, n) := 2^{c-1} \cdot n^{1 - (1/2^{c-1})} $.
The proof is by induction on the number of partitions. 

{\em Base case:} For $c = 2$, $f(c,n)$ becomes 
$2 \sqrt{n}$. Hence, the statement follows from Lemma~\ref{lem:twoPartitions}.

{\em Induction hypothesis:} The statement is true for any $c-1$ partitions. 
 
{\em Induction step:} Like in Lemma~\ref{lem:twoPartitions}, we divide the 
set of colors in $\P_1$ into two types of colors. 
Let $\P_1 = \{X_1, X_2, \dots, X_{r} \} \cup \{ Y_1, Y_2, \dots, Y_{q}\}$
such that $\abs{X_i} \geq \sqrt{n}$ and $\abs{Y_j} < \sqrt{n}$,
 for all $i \in [r], \; j \in [q]$.
Let us set $B_i = X_i$ and let $n_i := \abs{B_i}$, $\forall i \in [r]$ . 
Our base sets will be further subsets of these $B_i$s. 
For a fixed $i \in [r]$, let us define $\P'_h = \P_h|_{B_i} $, as a partition of the set $B_i$,
for all $h \in [c]$.
Clearly, $\P'_1$ has only one color. 
Now, we focus on the partition sequence $(\P'_2, \P'_3, \dots, \P'_c)$.
From the inductive hypothesis, 
there exists a partition of $B_i$ into $m_i$ base sets
$\{ B_{i,1}, B_{i,2}, \dots, B_{i,m_i} \}$ ($m_i \leq f(c-1, n_i)$) such that 
for any $u \in [m_i]$, there exists a permutation of 
$(\P'_2, \P'_3, \dots, \P'_c)$, given by $(\P'_{i_2}, \P'_{i_3}, \dots, \P'_{i_c})$,  
with $d_{B_{i,u}} (\P'_{i_2}, \P'_{i_3}, \dots, \P'_{i_c}) = 1$.
As $\P'_1$ has only one color, so does $\P'_1|_{B_{i,u}}$. 
Hence, $d_{B_{i,u}} (\P'_{i_2}, \P'_{i_3}, \dots, \P'_{i_c}, \P'_1)$ is also $1$.
From this, we easily get $d_{B_{i,u}} (\P_{i_2}, \P_{i_3}, \dots, \P_{i_c}, \P_1) = 1$.
The above argument can be made for all $i \in [r]$.

Now for the other colors, 
we proceed as in Lemma~\ref{lem:twoPartitions}.
Let $Y_j = \{y_{j,1}, y_{j,2}, \dots, y_{j, r_j} \}$, for all $j \in [q]$.
Let $r' = \max\{ r_1, r_2, \dots, r_q\}$ ($r' < \sqrt{n}$). 
Now define sets $B'_1, B'_2, \dots, B'_{r'}$ such that
for any $a \in [r']$, $B'_a = \{y_{j,a} \mid j \in [q], \; \abs{Y_j} \geq a \} $. 
In other words, all those $Y_j$s which have at least $a$ elements, contribute
their $a$-th element to $B'_a$.
Let $n'_a := \abs{B'_a}$, for all $a \in [r']$.
Our base sets will be further subsets of these $B'_a$s. 
For a fixed $a \in [r']$, let us define $\P'_h = \P_h|_{B'_a} $, 
as a partition of the set $B'_a$,
for all $h \in [c]$.
Clearly, $\P'_1$ has exactly one element in each of its colors. 
Now, we focus on the partition sequence $(\P'_2, \P'_3, \dots, \P'_c)$.
From the inductive hypothesis, 
there exists a partition of $B'_a$ into $m'_a$ base sets
$\{ B'_{a,1}, B'_{a,2}, \dots, B'_{a,m'_a} \}$ ($m'_a \leq f(c-1, n'_a)$) such that 
for any $u \in [m'_a]$, there exists a permutation of 
$(\P'_2, \P'_3, \dots, \P'_c)$, given by $(\P'_{i_2}, \P'_{i_3}, \dots, \P'_{i_c})$,  
with $d_{B'_{a,u}} (\P'_{i_2}, \P'_{i_3}, \dots, \P'_{i_c}) = 1$.
As $\P'_1$ has exactly one element in each of its colors, so does $\P'_1|_{B'_{a,u}}$. 
Hence, $d_{B'_{a,u}} (\P'_1, \P'_{i_2}, \P'_{i_3}, \dots, \P'_{i_c})$ is also $1$.
From this, we easily get $d_{B'_{a,u}} (\P_1, \P_{i_2}, \P_{i_3}, \dots, \P_{i_c}) = 1$.
The above argument can be made for all $a \in [r']$.

Our final set of base sets will be $\{B_{i,u} \mid i \in [r], \; u \in [m_i] \} \cup
\{ B'_{a,u} \mid a \in [r'], \; u \in [m'_a] \}$. 
As argued above, when restricted to any of these base sets, the given partitions 
have a sequence, which has distance $1$. 
Now, we need to bound the number of these base sets, 
$$m = \sum_{i \in [r]} m_i + \sum_{a \in [r']} m'_a .$$
From the bounds on $m_i$ and $m'_a$, we get
$$m \leq \sum_{i \in [r]} f(c-1, n_i) + \sum_{a \in [r']} f(c-1, n'_a) .$$
Recall that $n_i \geq \sqrt{n}$. 
We break the second sum, in the above equation,
into two parts.
Let $R_1 = \{a \in [r'] \mid n'_a \geq \sqrt{n} \}$ 
and $R_2 = \{ a \in [r'] \mid n'_a < \sqrt{n} \}$.
\begin{equation}
\label{eq:threeParts}
m \leq \sum_{i \in [r]} f(c-1, n_i) +  \sum_{a \in R_1} f(c-1, n'_a) + \sum_{a \in R_2} f(c-1, n'_a).
\end{equation}

Let us first focus on the third sum. 
Note that $\abs{R_2} \leq r' < \sqrt{n}$.
For $a \in R_2$, $n'_a < \sqrt{n}$ and hence 
$f(c-1, n'_a) < f(c-1, \sqrt{n}) = 
2^{c-2} \cdot n^{1/2 - (1/2^{c-1})}$.
So, 
\begin{equation}
\label{eq:sum3}
\sum_{a \in R_2} f(c-1, n'_a) < \sqrt{n} \cdot 2^{c-2} \cdot n^{1/2 - (1/2^{c-1})} 
= 2^{c-2} \cdot n^{1 - (1/2^{c-1})} .
\end{equation}

Now, we focus on first two sums in Equation~(\ref{eq:threeParts}). As, $n_i \geq \sqrt{n}, \; \forall i \in [r]$ and 
$n'_a \geq \sqrt{n}, \; \forall a \in R_1$, we combine these two sums (with an abuse of notation)
and write the sum as follows,
$$ \sum_{i \in [r'']} f(c-1, n_i) ,$$
where $r'' = r + \abs{R_1}$, and $n_i \geq \sqrt{n}, \; \forall i \in [r''] $. 
As each $n_i \geq \sqrt{n}$, we know $r'' < \sqrt{n}$ (as $\sum n_i \leq n$). 

Observe that $f(c-1, z)$, as a function of $z$, is a concave function (its derivative is monotonically decreasing, when $z > 0$). From the properties of a concave function, we know,
$$ \frac{1}{r''}\sum_{i \in [r'']} f(c-1, n_i) \leq 
       f \left( c-1, \frac{1}{r''} \sum_{ i \in [r'']} n_i \right).
$$
Now, $\sum_{i \in [r'']} n_i \leq n$ and $f(c-1, z)$ is an increasing function (when $z >0$). Hence, 
$$ \frac{1}{r''}\sum_{i \in [r'']} f(c-1, n_i) \leq 
       f \left( c-1, \frac{1}{r''} n \right).
$$
Equivalently,

\begin{eqnarray*} 
\sum_{i \in [r'']} f(c-1, n_i) &\leq &
      r''  \cdot 2^{c-2} \cdot (n/r'')^{1 - (1/2^{c-2})} \\
	&=& 2^{c-2} \cdot n^{1- (1/2^{c-2})} \cdot (r'')^{1/2^{c-2}} \\
	& < & 2^{c-2} \cdot n^{1- (1/2^{c-2})} \cdot n^{1/2^{c-1}} \\
	& = & 2^{c-2} \cdot n^{1- (1/2^{c-1})}. 
\end{eqnarray*}

Using this with Equation~(\ref{eq:sum3}) and substituting in Equation~(\ref{eq:threeParts}),
we get
$$ m < 2^{c-1} \cdot n^{1- (1/2^{c-1})} .$$
\end{proof}

Now, we combine these results with our hitting-sets for 
depth-$3$ circuits having $m$ base sets with $\di$-distance. 

\begin{reptheorem}{thm:c-setmultihs}
Let $C(\mathbf{x})$ be a $n$-variate polynomial, which can be computed by a sum of $c$ 
set-multinear depth-$3$ circuits, each having top fan-in $k$. 
Then there is a $(nck)^{O(2^{c-1} n^{1 - \epsilon} \log n) }$-time whitebox PIT test for $C$,
where $\epsilon := 1/2^{c-1}$.
\end{reptheorem}
\begin{proof}
As mentioned earlier, the polynomial $C(\mathbf{x})$ can be viewed as 
being computed by a depth-$3$ multilinear circuit, such that 
its product gates induce at most $c$-many distinct partitions. 
From Lemma~\ref{lem:cPartitions}, we can partition the variable set into
 $m$ base sets, such that for each of these base sets, the partitions can be sequenced
to have distance $1$, where $m := 2^{c-1} n^{1 - \epsilon}$.
Hence, the polynomial $C$ has $m$ base sets with $1$-distance and top fan-in $ck$.
Moreover, from the proof of Lemma~\ref{lem:cPartitions}, it is clear that 
such base sets can be computed in $n^{O(c)}$-time. 
From Theorem~\ref{thm:baseSetsHS}, we know that there is $(nck)^{O(m \log n)}$-time
whitebox PIT test for such a circuit. 
 Substituting
the value of $m$, we get the result. 
\end{proof}

\subsection*{Tightness of this method}
Lemma~\ref{lem:twoPartitions} can be put in other words as: 
Any two partitions have $m$-base-sets-$\di$-distance 
with $m \di = O(\sqrt{n})$. We can, in fact, show that this result is tight. 

\emph{Showing the lower bound:} Let $d(\P_1, \P_2) = \di$. Then each color of $\P_2$ has 
a friendly neighborhood (of at most $\di$ colors) which is exactly partitioned in
$\P_1$. Now construct $\di$ base sets such that $i$-th base set takes the variables 
of $i$-th color from every neighborhood of $\P_2$. Clearly, when restricted to one of these bases sets,
$d(\P_1,\P_2)$ is $1$. In other words $\P_1$ and $\P_2$ have $\di$-base-sets-$1$-distance.
Similarly, one can argue that if $\P_1$ and $\P_2$ have $m$-base-sets-$\di$-distance
then they also have $m \di$-base-sets-$1$-distance.
Now, we will show that if we want $m$-base-sets-$1$-distance for two partitions then $m = \Omega(\sqrt{n})$.

Consider the following example (assuming $n$ is a square): 

$\P_1 = \{\{1,2, \dots \sqrt{n}\}, 
\{\sqrt{n}+1, \sqrt{n} + 2, \dots, 2\sqrt{n} \}, 
\dots, \{\sqrt{n}(\sqrt{n}-1)+1, \sqrt{n}(\sqrt{n}-1)+2, \dots, n \} \}$ 
and 

$\P_2 = \{\{1,\sqrt{n}+1, \dots, n- \sqrt{n}+1\}, 
\{2,\sqrt{n}+2, \dots, n- \sqrt{n}+2\}, 
\dots, \{ \sqrt{n}, 2 \sqrt{n}, \dots, n\} \}$. 
Basically, $\P_2$ has the residue classes (mod $\sqrt{n}$).
\begin{observation}
A base set $B,$ such that $d_B(\P_1,\P_2) = 1$,
has at most $\sqrt{n}$ variables. 
\end{observation}
\begin{proof}
Suppose it has more than $\sqrt{n}$ variables.
Then, there is at least one color in $\P_1$ which
contributes two variables to $B$. These two variables have to be in two different colors of $\P_2$ 
(because of our design of $\P_1$ and $\P_2$).
So, $d_B(\P_1,\P_2)$ is at least $2$. We get a 
contradiction.
\end{proof}

The number of such base sets has to be at least $\sqrt{n}$.
Combining this with the reduction from $m$-base-sets-$\di$-distance
to $m \di$-base-sets-$1$-distance, we get $m \di = \Omega(\sqrt{n})$.

It is not clear if Lemma~\ref{lem:cPartitions} is tight. We conjecture
that for any set of partitions, $m \di = O(\sqrt{n})$ can be achieved.

\section{Sparse-Invertible Width-$w$ ROABP: Theorem~\ref{thm:invROABPHS}}
\label{sec:invROABP}
As mentioned in Section~\ref{sec:preliminaries},
a polynomial $C(\x)$ computed by $s$-sparse-factor width-$w$ ROABP
can be written as $D_0^{\top} (\prod_{i=1}^d D_i ) D_{d+1}  $,
where $D_i \in \F^{w \times w}[\x_i]$ is an $s$-sparse polynomial 
for all $i \in [d]$, and $\x_1,\x_2 \dots, \x_d$ are disjoint 
sets of variables. 

We will show a hitting-set for a sparse-factor ROABP 
$D_0 (\prod_{i=1}^{d} D_i ) D_{d+1} $ with 
$D_i$ being an {\em invertible} matrix, for all $i \in [d]$.
Hence, we name this model \emph{sparse-invertible-factor ROABP}.  
To be more general, we take $D_0$ and $D_{d+1}$ also to be polynomials
in some sets of variables disjoint from $\x_1, \x_2, \dots, \x_d$.

For a polynomial $D$, let its sparsity  
$\sp(D)$ be the number of monomials in $D$
with nonzero coefficients and
let $\mu(D)$ be the maximum support 
of any monomial in $D$.

\begin{reptheorem}{thm:invROABPHS}
\label{thm:ROABPHS-rep}
Let $\mathbf{x} = \mathbf{x}_0 \sqcup \dotsm \sqcup \mathbf{x}_{d+1}$, with $\abs{\mathbf{x}} = n$.
Let $C(\mathbf{x}) = D_0^{\top} D D_{d+1} \in \F[\mathbf{x}]$ be a polynomial
with $D(\mathbf{x}) = \prod_{i=1}^{d} D_i(\mathbf{x}_i)$, where
$D_0 \in \F^w[\mathbf{x}_0]$ and $D_{d+1} \in \F^w[\mathbf{x}_{d+1}]$
and for all $i \in [d]$, 
$D_i \in \F^{w \times w}[\mathbf{x}_i]$ is an invertible matrix.
For all $i \in \{0, 1, \dots, d+1 \}$, $D_i$ has degree bounded
by $\delta$, $\sp(D_i) \leq s$ and $\mu(D_i) \leq \mu$.
Let $\ell := 1 + 2 \min\{ \ceil{\log (w^2 \cdot s)}, \mu \}$.
Then there is a hitting-set of size 
$\poly((n \delta s)^{\ell w^2})$ for $C(\mathbf{x})$.
\end{reptheorem}

\begin{remark}
If $\mu=1$, e.g.\ each $D_i$ is either a univariate or a linear polynomial, then
we get poly-time for constant $w$. Also if both $w$
and the sparsity-bound $s$ are constant, we get poly-time.
\end{remark}

Like \cite{ASS13} and \cite{FSS13}, we find a hitting-set by showing
a {\em low-support concentration}. 
Low support concentration in the polynomial $D(\x) = \prod_{i=1}^{d} D_i $
means that the coefficients of the low support monomials in $D(\x)$ span the whole
coefficient space of $D(\x)$. 

Let $\x$ be $\{x_1, x_2 \dots, x_n\}$.
For any $e \in \Zp^n$, 
support of the monomial $\mathbf{x}^e$ is defined as 
$\Supp(e) := \{i \in [n] \mid e_i \neq 0 \}$
and support size is defined as $\suppo(e) := \abs{\Supp(e)}$.
Now, we define \emph{$\ell$-concentration} for a polynomial 
$ D(\x) \in \F^{w \times w}[\x]$.

\begin{definition}[$\ell$-concentration]
Polynomial $ \Dx \in \F^{w \times w}[\x]$ is $ \ell $-concentrated if
$
\rank_\F \{\coeff_D (\x^e) \mid e \in \Zp^n, \; \suppo(e) < \ell\} 
= \rank_\F \{\coeff_D (\x^e) \mid e \in \Zp^n\}.
$
\end{definition}

We will later see that the low support concentration in
polynomial $D(\x)$ implies low support concentration in polynomial $C(\x)$ 
(defined similarly).
In other words, $C(\x)$ will have 
a nonzero coefficient for at least one of the low support monomials. 
Thus, we get a hitting set by testing these low support coefficients.
We use the following lemma from \cite{ASS13}.

\begin{lemma}
\label{lem:hsFromlConc}
If $C(\x) \in \F[\x]$ is an $n$-variate, $\ell$-concentrated polynomial 
with highest individual degree $\degree$, 
then there is a $(n \degree)^{O(\ell)}$-time hitting-set for $C(\x)$.
\end{lemma}
\begin{proof}
$\ell$-concentration for $C(\x)$ simply means that 
it has at least one $(< \ell)$-support monomial
with nonzero coefficient. 
We will construct a hitting set which essentially will test all these 
$(< \ell)$-support coefficients. 
We go over all subsets $S$ of $\x$ with size $\ell-1$ and do the following:
Substitute $0$ for all the variables outside the set $S$.
There will be at least one choice of $S$, for which the polynomial $C(\x)$
remains nonzero after the substitution.
Now, it is an $(\ell -1)$-variate nonzero polynomial. 
We take the usual hitting set $\Hit^{\ell-1}$ for this, where
$\Hit \subseteq \F$ is a set of size $\degree +1$ (see, for example, \cite[Fact 4.1]{SY10}).
In other words, each of these $\ell -1 $ variables are assigned values from the set $\Hit$.

The number of sets $S$ we need to try are ${n}\choose{\ell -1}$. 
Hence, the overall hitting set size is $(n \degree)^{O(\ell)}$.
\end{proof}

Now, we move on to show how to achieve low support 
concentration in $ D(\x) = \prod_{i=1}^{d} D_i $.
To achieve that we will use some efficient shift.
By shifting by a point $\mathbf{\alpha} := (\alpha_1, \alpha_2, \dots, \alpha_n)$,
we mean replacement of $x_i$ with $x_i + \alpha_i$.
Note that $D(\x + \mathbf{\alpha} ) \neq 0$
if and only if $D(\x) \neq 0$.
Hence, a hitting set for $D(\x + \mathbf{\alpha})$
gives us a hitting set for $D(\x)$. 
Instead of constants, we will be actually shifting $D(\x)$ by univariate polynomials,
say, given by the map $\phi \colon \mathbf{t} \to \{t^a\}_{a \geq 0} $, where
$\mathbf{t} := \{t_1, t_2, \dots, t_n \} $.
The $\phi$ is said to be an efficient map if  
 $\phi(t_i)$ is efficiently computable, for each $i \in [n]$.

\noindent {\bf Proof Idea-} 
As all the matrices in the matrix product $D(\mathbf{x}) = \prod_{i=1}^{d} D_i(\mathbf{x}_i)$
are over disjoint sets of variables,
any coefficient in the polynomial $D(\mathbf{x})$ can be uniquely written as a product of 
$d$ factors, each coming from one $D_i$. 
We start with the assumption that the constant term of each polynomial $D_i$, 
denoted by $D_{i {\bf 0}}$, is an invertible matrix.
Using this we define a notion of {\em parent} and {\em child}
between all the coefficients (also see Figure~\ref{fig:tree}): 
If a coefficient can be obtained
from another coefficient by replacing one of its constant factors $D_{i {\bf 0}}$
with another term (with non-trivial support) from $D_i$, 
then former is called a parent
of the latter. 
Observe that if we want to do this replacement by a multiplication of some matrix, 
then $D_{i {\bf 0}}$ should be invertible.
Moreover, all the factors on its right side (or its left side) 
also need to be
constant terms in their respective matrices (this is because of non-commutativity).
For a coefficient, the set of matrices $D_i$ which contribute a non-trivial
factor to it, is said to form the {\em block-support} of the coefficient.

Our next step is to show that if a coefficient linearly depends on its descendants 
then the dependence can be lifted to its parent (by dividing and multiplying appropriate factors), i.e.\ its parent also linearly depends 
on its descendants. 
As the dimension of the matrix algebra is constant, 
if we take an appropriately large (constant) child-parent chain,
there will be a linear dependence among the coefficients in the chain.
As the dependencies lift to the parent, they can be lifted all the way up. 
By an inductive argument it follows that every coefficient depends on 
the coefficients with low-block-support. 
Now, this can be translated to low-support concentration in $D$,
if a low-support concentration is assumed in each $D_i$. 

To achieve low-support concentration in each $D_i$,
we use an appropriate shift. The sparsity of $D_i$ is used crucially 
in this step. 
To make $D_{i {\bf 0}}$ invertible, again an appropriate shift is used. 
Note that $D_{i {\bf 0}}$ can be made invertible by a shift only when 
$D_i$ itself is invertible, hence the invertible-factor assumption.

\subsection{Building the Proof of Theorem~\ref{thm:invROABPHS}}
\label{sec:proofinvROABP}
Our first focus will be on
the matrix product $D(\mathbf{x}) := \prod_{i=1}^{d} D_i$
which belongs to $\F^{w \times w}[\mathbf{x}]$. 
We will show low-support concentration in $D(\mathbf{x})$
over the matrix algebra $\F^{w \times w}$ (which is non-commutative!).


\subsubsection{Low Block-Support} 
Let the matrix product $D(\mathbf{x}) := \prod_{i=1}^{d} D_i$
correspond to an ROABP such that
$D_i \in \F^{w \times w}[\mathbf{x}_i]$ for all $i \in [d]$.
Let $n_i$ be the cardinality of $\mathbf{x}_i$ and
let $n = \sum_{i=1}^d n_i$. 
For an exponent $e = (e_1, e_2, \dots, e_m) \in \Zp^m$, 
and for a set of variables $\mathbf{y}= \{y_1,y_2, \dots, y_m\}$,
${\mathbf{y}}^e$ will denote ${{y}_1}^{e_1} {y_2}^{e_2} \dots {y_m}^{e_m}$.

Viewing $D_i$ as belonging to $\F^{w \times w}[{\mathbf{x}}_i]$,
one can write $D_i := \sum_{e \in \Zp^{n_i}} D_{ie} {\mathbf{x}}_i^e$,
where $D_{ie} \in \F^{w \times w}$, for all $e \in \Zp^{n_i}$.
In particular $D_{i {\bf 0}}$ refers to
the constant part of the polynomial $D_i$. 

For any $e \in \Zp^n$, 
support of the monomial $\mathbf{x}^e$ is defined as 
$\Supp(e) := \{i \in [n] \mid e_i \neq 0 \}$
and support size is defined as $\suppo(e) := \abs{\Supp(e)}$.
In this section, we will also define block-support
of a monomial. 
Any monomial $\mathbf{x}^e$ for $e \in \Zp^n$, 
can be seen as a product $\prod_{i=1}^d {\mathbf{x}_i}^{e_i}$,
where $e_i \in \Zp^{n_i}$ for all $i \in [d]$,
such that $e = (e_1, e_2, \dots, e_d)$.
We define {\em block-support} of $e$, $\bS(e)$  as
$\{i \in [d] \mid e_i \neq {\bf 0} \}$ and 
{\em block-support size} of $e$, $\bs(e) = \abs{\bS(e)}$.
 
Next, we will show low block-support concentration 
of $D(\mathbf{x})$ when {\em each $D_{i {\bf 0}}$ is invertible}. 

As each $D_i$ is a polynomial over a different set of variables,
we can easily see that the coefficient of any 
monomial $\xe = \prod_{i=1}^{d} \xei$ in $D(\mathbf{x})$
is 
\begin{equation}
\label{eq:coeffSplits}
D_e := \prod_{i=1}^d D_{ie_i}.
\end{equation}
Now, we will define a relation of {\em parent} and {\em children} 
between these coefficients. 

\begin{definition}
For $e^*,e \in \Zp^n$, $D_{e^*}$ is called a parent of $D_{e}$
if 
$\exists j \in [d]$,  
$j > \max \bS(e)$ or $j < \min  \bS(e) $,
such that
$\bS(e^*) = \bS(e) \cup \{j\}$
and $e^*_i = e_i$, $\forall i \in [d]$ with $i \neq j$.
\end{definition}

\begin{figure}
\centering
%
  \input{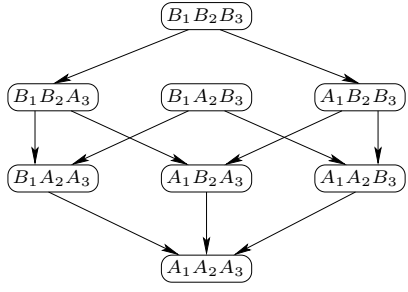}%

\caption{An edge represents the child-parent relationship among the coefficients.
The arrow points towards the child.}
\label{fig:tree}
\end{figure}

If $D_{e^*}$ is a parent of $D_{e}$ then 
$D_{e}$ is a {\em child} of $D_{e^*}$.
Note that {\em a coefficient has at most two children},
on the other hand it can have many parents.
In the case when $j > \max \bS(e)$
 we call $e$, the {\em left} child of $e^*$ and 
in the other case we call it the {\em right} child.
Figure~\ref{fig:tree} shows this relationship between the coefficients 
for the polynomial $(A_1 + B_1 x_1)(A_2 + B_2 x_2)(A_3 + B_3 x_3)$,
where $A_i,B_i \in \F^{w \times w}$, for all $i \in \{1,2,3\}$.

To motivate this definition, 
observe that if $j > \max  \bS(e)$ then by Equation~(\ref{eq:coeffSplits}) 
we can write
$D_{e^*} = D_e A^{-1} B$, where $A := \prod_{i=j}^{d} D_{i{\bf 0}}$
and $B := D_{j e^*_{j}}\prod_{i=j+1}^{d} D_{i{\bf 0}}$.
We will denote the product $A^{-1}B$ as $D_{e^{-1}e^*}$.
Similarly, if $j < \min \bS(e)$ then one can write
$D_{e^*} = B A^{-1} D_e $, where $A := \prod_{i=1}^{j} D_{i{\bf 0}}$
and $B := \left(\prod_{i=1}^{j-1} D_{i{\bf 0}} \right) D_{j e^*_{j}}$.
In this case we will denote the product $B A^{-1}$ as $D_{e^*e^{-1}}$.
Note that the invertibility of $D_{i {\bf 0}}$s is crucial here. 

We also define {\em descendants} of a coefficient $D_e$
as $\descend(D_e) := \{D_f \mid f \in \Zp^n, \; \bS(f) \subset \bS(e) \}$.
Note that, the set of descendants of a coefficient 
could be bigger than the set of its children, grand-children, etc.
Now, we will view the coefficients as $\F$-vectors 
and look at the linear dependence between them. 
The following lemma shows how these dependencies lift to the parent. 

\begin{lemma}[Child to parent]
Let $D_{e^*}$ be a parent of $D_{e}$. 
If $D_{e}$ is linearly dependent on its descendants,
then $D_{e^*}$ is linearly dependent on its descendants. 
\label{lem:parentchild}
\end{lemma} 
\begin{proof}
Let $D_{e}$ be the left child of $D_{e^*}$ 
(the other case is similar).
So, we can write 
\begin{equation}
D_{e^*} =  D_{e} D_{e^{-1}e^*}.
\label{eq:parentchild}
\end{equation}
Let the dependence of $D_e$ on its descendants be the following: 
$$
D_{e} = \sum_{\substack{f \\ \bS(f) \subset \bS(e)}} \alpha_f D_{f}.
$$
Using Equation~(\ref{eq:parentchild}) we can write,
$$
D_{e^*} = \sum_{\substack{f \\ \bS(f) \subset \bS(e)}} \alpha_f D_{f} D_{e^{-1}e^*}.
$$
Now, we just need to show that
for any $D_f$ with $\bS(f) \subset \bS(e)$,
 $D_{f} D_{e^{-1}e^*}$ 
is a valid coefficient of some monomial in $\Dx$ and 
also that it is a descendant of $D_{e^*}$.
Recall that $D_{e^{-1}e^*} = A^{-1}B$, 
where $A := \prod_{i=j}^{d} D_{i{\bf 0}}$
and $B := D_{j e^*_{j}}\prod_{i=j+1}^{d} D_{i{\bf 0}}$
and $\bS(e^*) = \bS(e) \cup \{j\}$.
We know that $j > \max\{\bS(e)\}$. 
Hence, $j > \max\{\bS(f)\}$ as $\bS(f) \subset \bS(e)$.
So, it is clear that $D_{f} D_{e^{-1}e^*}$
is the coefficient of ${\mathbf{x}}^{f^*} := {\mathbf{x}}^f {\mathbf{x}_{j}}^{e^*_{j}}$.
It is easy to see that $\bS(f^*) = \bS(f) \cup \{j\} \subset \bS(e^*)$.
Hence, $D_{f^*}=D_{f} D_{e^{-1}e^*} $ is a descendant of $D_{e^*}$.

\end{proof}

Clearly, if the descendants are more than $\dim_{\F}\F^{w \times w}$,
then there will be a linear dependence among them. So, 

\begin{lemma}
Any coefficient $D_e$, with $\bs(e) = w^2$, $\F$-linearly depends on
its descendants. 
\label{lem:blockSupportk}
\end{lemma}
\begin{proof}
First of all, we show that if a 
coefficient $D_{f^*}$ is nonzero then so are
its children. 
Let us consider its left child $D_f$ (the other case is similar).
Recall that we can write $D_{f^*} = D_f D_{f^{-1}f^*}$.
Hence if $D_f$ is zero, so is $D_{f^*}$. 

Let $k:= w^2$.
Now, consider a chain of coefficients
$D_{e_0}, D_{e_1}, \dots, D_{e_k}=D_e$,
such that for any $i \in [k]$, 
$D_{e_{i-1}}$ is a child of $D_{e_{i}}$. 
Clearly, $\bs(e_i) = i$ for $0 \leq i \leq k$. 
All the vectors in this chain are nonzero
because of our above argument, as $D_e$ is nonzero
(The case of $D_e = 0$ is trivial).
These $k+1$ vectors lie in $\F^{k}$, 
hence, there exists an $i \in [k]$
such that $D_{e_{i}}$ is linear dependent on 
$\{ D_{e_0}, \dots, D_{e_{i-1}}\}$. 
As descendants include children, grand-children, etc.,
we can say that $D_{e_i}$ is linearly dependent on
its descendants. 
Now, by applying Lemma~\ref{lem:parentchild}
repeatedly, we conclude $D_{e_{k}} = D_e$ 
is dependent on its descendants. 
\end{proof}

Note that, for a coefficient $D_e$ with $\bs(e) = i$,
its descendants have block-support strictly smaller than
$i$.
So, Lemma~\ref{lem:blockSupportk} means that coefficients
with block-support $w^2$ depend on coefficients with block-support
$\leq w^2-1$.
Now, we show $w^2$-block-support-concentration in $\Dx$,
i.e.\ any coefficient is   
dependent on the coefficients with block-support $\leq w^2-1$.

\begin{lemma}[$w^2$-Block-concentration]
\label{lem:bSConc}
Let $\Dx = \prod_{i=1}^{d} D_i (\mathbf{x}_i) \in \F^{w \times w}[\mathbf{x}]$
 be a polynomial
with $D_{i\bf{0}}$ being invertible for each $i \in [d]$.
Then $\Dx$ has $w^2$-block-support concentration. 
\end{lemma} 
\begin{proof}
Let $k:= w^2$. We will actually show that
for any coefficient $D_e$ with $\bs(e) \geq k$ (the case when $\bs(e) < k$ is trivial), 
$$D_e \in \Span\{ D_f \mid f \in \Zp^n, \; \bS(f) \subset \bS(e) \text{ and } \bs(f) \leq k-1 \}.$$
We will prove the statement by induction on
the block-support of $D_e$, $\bs(e)$.

{Base case:} When $\bs(e)=k$, it has been already shown in 
Lemma~\ref{lem:blockSupportk}.

{\em Induction Hypothesis:} For any coefficient $D_e$
with $\bs(e) = i-1$ for $i-1 \geq k$, 
$$D_e \in \Span\{ D_f \mid f \in \Zp^n, \; \bS(f) \subset \bS(e) \text{ and } \bs(f) \leq k-1 \}.$$

{\em Induction step:}
Let us take a coefficient $D_e$ with $\bs(e) = i$. 
Consider any child of $D_e$, denoted by $D_{e'}$. 
As $\bs(e')=i-1$, by our induction hypothesis, 
$D_{e'}$ is linearly dependent on its descendants. 
So, from Lemma~\ref{lem:parentchild}, 
$D_{e}$ is linearly dependent on its descendants. 
In other words, 
\begin{equation}
D_e \in \Span\{ D_f \mid \bS(f) \subset \bS(e) \text{ and } \bs(f) \leq i-1 \}.
\label{eq:De}
\end{equation}
Again, by our induction hypothesis, for any coefficient $D_f$, 
with $\bs(f) \leq i-1$, 
\begin{equation}
D_f \in \Span\{ D_g \mid \bS(g) \subset \bS(f) \text{ and } \bs(g) \leq k-1 \}.
\label{eq:Df}
\end{equation}
Combining Equations (\ref{eq:De}) and (\ref{eq:Df}),
we get
$$D_e \in \Span\{ D_g \mid \bS(g) \subset \bS(e) \text{ and } \bs(g) \leq k-1 \}.$$
\end{proof}

Now, we show low block-support concentration in
the actual polynomial computed by an ROABP, i.e.\
in $C(\mathbf{x}) = D_0^{\top} (\prod_{i=1}^d D_i) D_{d+1}$, 
where $D_0,D_{d+1} \in F^w[\mathbf{x}]$.
Note that in context of $C$, the definition of block support
is appropriately modified. Block support of a monomial now
is a subset of $\{0,1, \dots, d+1\}$. As before it will contain
the index $i$, if the monomial has a non-trivial support from $\x_i$,
for $0 \leq i \leq d+1$.

\begin{lemma}
\label{lem:CbSConc}
Let $\mathbf{x} = \mathbf{x}_0 \sqcup \mathbf{x}_1 \sqcup \dotsm \sqcup \mathbf{x}_{d+1}$.
Let $D(\o{x}) \in \F^{w \times w}[\mathbf{x}_1, \dots, \mathbf{x}_d]$
 be a polynomial described in Lemma~\ref{lem:bSConc}.
Let $C(\mathbf{x}) = D_0^{\top} D D_{d+1} \in \F[\mathbf{x}]$
 be a polynomial
with $D_0 \in \F^w[\mathbf{x}_0]$, $D_{d+1} \in \F^w[\mathbf{x}_{d+1}]$.
Then $C(\mathbf{x})$ has $(w^2+2)$-block-support concentration.  
\end{lemma}
\begin{proof}
Let $k:=w^2$.
Lemma~\ref{lem:bSConc} shows that $\Dx$
 has $k$-block-support concentration. 
The coefficient of $\mathbf{x}^e$ in $C$ 
is $C_e := D_{0e_0} \prod_{i=1}^d D_{ie_i} D_{(d+1)e_{d+1}}$,
where $e = (e_0, e_1, \dots,e_d, e_{d+1})$. 
Let $D_e := \prod_{i=1}^d D_{ie_i}$.
By $k$-block-support concentration of $\Dx$,
$$D_e \in \Span\{ D_f \mid \bs(f) \leq k-1 \}.$$
Which implies, 
$$C_e \in \Span\{ D_{0e_0} D_f D_{(d+1)e_{d+1}} \mid \bs(f) \leq k-1 \}.$$ 
Clearly, $D_{0e_0} D_f D_{(d+1)e_{d+1}}$ is the coefficient of the monomial 
$x_0^{e_0} x_1^{f_1} \dotsm x_d^{f_d} x_{d+1}^{e_{d+1}}$. 
Hence, $C_e \in \Span\{  C_f \mid \bs(f) \leq k+1 \}$.
\end{proof}

\subsection{Low-support concentration}
Now, we want to show that if 
$C(\mathbf{x}) = D_0^{\top} (\prod_{i=1}^{d} D_i)  D_{d+1}$ has
low block-support concentration and moreover if each $D_i$ has 
low-support concentration then $C(\x)$ has an appropriate low-support concentration. 

\begin{lemma}[Composition]
Let $C(\x)$ be a polynomial $D_0^{\top} D  D_{d+1}$
as described in Lemma~\ref{lem:CbSConc}. 
If $C(\x)$ has $\ell$-block-support concentration and
$D_i(\mathbf{x}_i)$ has $\ell'$-support concentration for all $i \in [d]$ then 
$C(\x)$ has $\ell \ell'$-support concentration.  
\label{lem:ll'}
\end{lemma}
\begin{proof}
Recall that as $D_i$'s are polynomials over disjoint sets of variables,
any coefficient $C_f$ in $C(\x)$ can be written as
$D_{0 f_0}^{\top} (\prod_{i=1}^d D_{if_i}) D_{(d+1)f_{d+1}}$, 
where $f = (f_0, f_1, f_2, \dots, f_{d+1})$ and
$D_{if_i}$ is the coefficient corresponding to the monomial $\mathbf{x}_i^{f_i}$
in $D_i$ for all $0 \leq i \leq d+1$.
From the definition of $\bS(f)$, we know that
$f_i = 0$, for any $i \notin \bS(f)$.
From $\ell'$-support concentration of $D_i(\mathbf{x}_i)$, 
we know that for any coefficient $D_{if_i}$,
$$D_{if_i} \in \Span \{D_{ig_i} \mid g_i \in \Zp^{n_i}, \; \suppo(g_i) \leq \ell'-1 \}.$$
Using this, we can write
\begin{equation}
\begin{split}
C_f \in \Span \left\{ D_{0g_0}^{\top} \prod_{i = 1}^d D_{ig_i} D_{(d+1)g_{d+1}} \right. & \mid  
 g_i \in \Zp^{n_i}, \; \suppo(g_i) \leq \ell'-1, \; \forall i \in [[d+1]]  \\ 
& \left. \text{ and } g_i={\bf 0}, \; \forall i \notin \bS(f) \vphantom{\prod_{i = 1}^d} \right\} .
\end{split}
\end{equation}
Note that the product $D_{0g_0}^{\top} \prod_{i = 1}^d D_{ig_i} D_{(d+1)g_{d+1}}$ 
will be the coefficient of a monomial $\mathbf{x}^g$ such that $\bS(g) \subseteq \bS(f)$
because $g_i={\bf 0}, \; \forall i \notin \bS(f)$.
Clearly, if $\suppo(g_i) \leq \ell'-1, \; \forall i \in \bS(f)$ then $\suppo(g) \leq (\ell'-1)\bs(f)$.
So, one can write
\begin{equation}
C_f \in \Span\{ C_g \mid g \in \Zp^n, \; \suppo(g) \leq (\ell'-1)\bs(f)  \} .
\label{eq:span-bsf}
\end{equation}
From $\ell$-block-support concentration of $C(\x)$,
 we know that for any 
coefficient $C_e$ of $C(\x)$, 
\begin{equation}
C_e \in \Span\{ C_f \mid f \in \Zp^n, \; \bs(f) \leq \ell - 1 \}.
\label{eq:span-l-1}
\end{equation}
Using Equations (\ref{eq:span-bsf}) and (\ref{eq:span-l-1}),
we can write for any 
coefficient $C_e$ of $C(\x)$, 
$$C_e \in \Span\{ C_g \mid g \in \Zp^n, \; \suppo(g) \leq (\ell'-1)(\ell - 1) \}.$$
Hence, $C(\x)$ has $((\ell - 1)(\ell'-1)+1)$-support concentration and
hence $\ell \ell'$-support concentration.
\end{proof}

Now, we just need to show low-support concentration 
of each $D_i$. 
To achieve that we will use some efficient shift.
Shifting will serve a dual purpose. 
Recall that for Lemma~\ref{lem:bSConc}, we need
invertibility of the constant term in $D_i$, i.e.\ $D_{i\bf{0}}$, 
for all $i \in [d]$.
In case $D_{i \bf{0}}$ is not invertible for some $i \in [d]$, 
after a shift it might become invertible, 
since $D_i$ is assumed invertible in the sparse-invertible model. 
For the shifted polynomial $D'_i(\mathbf{x}_i) := D_i(\mathbf{x}_i + \phi(\mathbf{t}_i))$,
its constant term $D'_{i\bf{0}}$ is just an evaluation of $D_i(\mathbf{x})$, 
i.e.\ $D_i|_{\mathbf{x}_i = \phi(\mathbf{t}_i)}$.
Now, we want a shift for $D_i$ which would ensure that 
$\det(D'_{i{\bf 0}}) \neq 0$ and that $D'_i$ has low-support concentration. 
For both the goals we use the sparsity of the polynomial. 

For a polynomial $D$, let its sparsity set 
$\Sp(D)$ be the set of monomials in $D$
with nonzero coefficients 
and $\sp(D)$ be its
sparsity, i.e.\ $\sp(D) = \abs{\Sp(D)}$.
Let,
for a polynomial $D(\mathbf{x}) \in \F^{w \times w}[\mathbf{x}]$,
$S = \Sp(D)$ and $s = \abs{S}$. 
Then it is easy to see that 
for its determinant polynomial 
$\Sp(\det(D)) \subseteq S^w$,
where $S^w := \{ m_1 m_2 \dotsm m_w \mid m_i \in S, \; \forall i \in [w]   \}$.
Hence $\sp(\det(D)) \leq s^w$.
Now, suppose $\det(D) \neq 0$. 
We will describe an efficient shift which will make the
constant term, of the shifted polynomial, invertible. 
Let $\phi \colon \mathbf{t} \to \{t^i\}_{i=0}^{\infty}$ be a monomial map
which separates all the monomials in $\det(D(\mathbf{t}))$,
i.e.\
for any two $\mathbf{t}^{e_1},{\mathbf{t}}^{e_2} \in \Sp(\det(D(\mathbf{t})))$,
$\phi(\mathbf{t}^{e_1}) \neq \phi(\mathbf{t}^{e_2})$. 
It is easy to see that if we shift each $x_i$ by $\phi(t_i)$
to get $D'(\mathbf{x}) = D(\mathbf{x} + \phi(\mathbf{t}))$ then 
$\det(D'_{i{\bf 0}}) = \det(D|_{\mathbf{x} = \phi(\mathbf{t})}) \neq 0$.

For sparse polynomials, Agrawal et al.\ \cite[Lemma 16]{ASS13} have 
given an efficient shift to achieve low-support concentration. 
Here, we rewrite their lemma. 
The map $\phi_{\ell'} \colon \mathbf{t} \to \{t^i\}_{i=0}^{\infty}$ is 
said to be separating $\ell'$-support monomials of degree $\delta$,
if for any two monomials $\mathbf{t}^{e_1}$ and $\mathbf{t}^{e_2}$ which have
support bounded by $\ell'$ and degree bounded by $\delta$, 
$\phi_{\ell'}(\mathbf{t}^{e_1}) \neq \phi_{\ell'}(\mathbf{t}^{e_2})$. 
For a polynomial $\Dx$, let $\mu(D)$ be the maximum support 
of a monomial in $D$, i.e.\ $\mu(D) := \displaystyle\max_{\mathbf{x}^e \in \Sp(D)} \suppo(e)$.

\begin{lemma}[\cite{ASS13}]
\label{lem:sparse}
Let $V$ be a $\F$-vector space of dimension $k$.
Let $D(\mathbf{x}) \in V[\mathbf{x}]$ be a polynomial with
degree bound $\delta$. 
Let $\ell := 1 + 2 \min\{ \ceil{\log (k \cdot \sp(D))}, \mu(D) \}$ and $\phi_{\ell}$ be a
monomial map separating $\ell$-support monomials of degree $\delta$.
Then $D(\mathbf{x}+\phi_{\ell}(\mathbf{t}))$ has $\ell$-concentration
over $\F(t)$.
\end{lemma} 

The \cite{ASS13} version of the Lemma~\ref{lem:sparse} gave
 a concentration result
about sparse polynomials over $\H_k(\F)$. But observe that
the process of shifting and the definition of concentration 
only deal with the additive structure of $\H_k(\F)$,
and the multiplication structure is irrelevant. Hence,
the result is true over any
$\F$-vector space,
in particular,
over the matrix algebra.  
By combining these observations, we have the following. 

\begin{lemma}
\label{lem:ROABPconc}
Let $D(\mathbf{x}) = \prod_{i=1}^{d} D_i(\mathbf{x}_i)$ be a polynomial 
in $\F^{w \times w}[\mathbf{x}]$ with $\det(D) \neq 0$
such that for all $i \in [d]$,
 $D_i$ has degree bounded
by $\delta$, $\sp(D_i) \leq s$ and $\mu(D_i) \leq \mu$.
Let $\ell := 1 + 2 \min\{ \ceil{\log (w^2 \cdot s)}, \mu \}$
and $M :=  \poly(s^w (n \delta)^{\ell})$.
Then there is a set of $M$ monomial maps with degree bounded by
$M \log M$ 
such that for at least one of the maps $\phi$,
$C' := C(\mathbf{x}+\phi(\mathbf{t}))$ has $\ell (w^2+2)$-concentration. 
\end{lemma}
\begin{proof}
Let $\phi \colon \mathbf{t} \to \{t^i\}_{i=0}^{\infty}$ be a map 
such that it separates all the monomials in $\Sp(\det(D_i(\mathbf{t}_i)))$,
for all $i \in [d]$. There are $ds^{2w}$ such monomial pairs. 
Also assume that $\phi$ separates all monomials of support bounded by $\ell$. 
There are $(n \delta)^{O(\ell)}$ such monomials. 
Hence, total number of 
monomial pairs which need to be separated are $s^{O(w)} + (n \delta)^{O(\ell)}$.
From Lemma~\ref{lem:kronecker}, 
we know that there is a set of $M$ monomial maps ($t_i \mapsto t^{w(t_i)}$)
with highest degree $M\log M$ such that at least one of the maps
$\phi$ separates the desired monomials, 
where $M = \poly(s^w (n \delta)^{\ell})$.
As the map $\phi$ separates all the monomials in $\Sp(\det(D_i(\mathbf{t}_i)))$,
$\det(D_i(\phi(\mathbf{t}_i))) \neq 0$ and hence, 
$D'_{i{\bf 0}}$ is invertible for all $i \in [d]$.
So, $C'(\mathbf{x})$ has $(w^2+2)$-block-support concentration from Lemma~\ref{lem:bSConc}.

From Lemma~\ref{lem:sparse}, $D'_i(\mathbf{x}_i)$ has $\ell$-concentration for all 
$0 \leq i \leq d+1$.
Hence, from Lemma~\ref{lem:ll'}, $C'(\mathbf{x})$ has $\ell (w^2+2)$-concentration.  
\end{proof}

Now, we come back to the proof of Theorem~\ref{thm:invROABPHS} (restated in this section).
Combining Lemma~\ref{lem:ROABPconc} with Lemma~\ref{lem:hsFromlConc} we get a hitting set
for $C'(\x) = C(\x + \phi(\mathbf{t}))$ of size $(n \delta)^{O(\ell w^2)}$.
Each of these evaluations of $C$ will be a polynomial in $t$ 
with degree at most $\poly(s^w (n \delta)^{\ell})$.
Hence, total time complexity becomes $\poly(s^w(n \delta)^{\ell w^2})$. 

\comment{
Note that for constant width ROABP,
when $\mu(D_i)$ is bounded by a constant for each $0 \leq i \leq d+1$,
in particular when each $D_i$ is univariate, 
the parameter $\ell$ becomes constant and
the hitting-set becomes polynomial-time.  
}

\subsection{Width-$2$ Read Once ABP}
\label{sec:2ROABP}
In the previous section, the crucial part in finding a hitting-set for an ROABP,
is the assumption that the matrix product $\Dx$ is invertible. 
Now, we will show that for width-$2$ ROABP, this assumption is not required.
Via a factorization property of $2 \times 2$ matrices,
 we will show that PIT for width-$2$ sparse-factor ROABP
reduces to PIT for width-$2$ sparse-invertible-factor ROABP. 

\begin{lemma}[$2 \times 2$ invertibility]
\label{lem:width2}
Let $C(\mathbf{x}) =  D_0^{\top} \left( \prod_{i=1}^d D_i  \right) D_{d+1}$ be a polynomial computed
by a width-$2$ sparse-factor ROABP.
Then we can write $\alpha(\mathbf{x}) C(\mathbf{x}) = C_1(\mathbf{x})C_2(\mathbf{x}) \dotsm C_{m+1}(\mathbf{x})$,
for some nonzero $\alpha \in \F[\mathbf{x}]$ and some $m \leq d$, 
where $C_i(\mathbf{x})$ is a polynomial computed by a width-$2$ sparse-invertible-factor 
ROABP, for all $i \in [m+1]$. 
\end{lemma}
\begin{proof}
Let us say, for some $i \in [d]$, $D_i(\mathbf{x}_i)$ is not invertible. 
Let $D_i = \left[ \begin{smallmatrix} a_i & b_i \\ c_i & d_i \end{smallmatrix} \right]$ with
$a_i,b_i,c_i,d_i \in \F[\mathbf{x}_i]$ and $a_i d_i = b_i c_i$.
Without loss of generality, at least one of $\{a_i,b_i,c_i,d_i\}$ is nonzero. 
Let us say $a_i \neq 0$ (other cases are similar).
Then we can write,
$$ \begin{bmatrix} a_i & b_i \\ c_i & d_i \end{bmatrix} 
= \frac{1}{a_i} \begin{bmatrix} a_i \\ c_i  \end{bmatrix}
\begin{bmatrix} a_i & b_i \end{bmatrix} . 
$$ 
In other words, we can write
$\alpha_i D_i = A_i B_i^{\top} $, where $A_i,B_i \in \F^2[\mathbf{x}_i]$ and 
$0 \neq \alpha_i \in \{a_i,b_i,c_i,d_i\} $.
Note that $\sp(\alpha_i), \sp(A_i),\sp(B_i) \leq \sp(D_i)$.
Let us say that the set of non-invertible $D_i$s is $\{D_{i_1}, D_{i_2}, \dots, D_{i_{m}} \}$.
Writing all of them in the above form we get, 
$$C(\mathbf{x}) \prod_{j=1}^m \alpha_{i_j} = 
\prod_{j=1}^{m+1} C_j ,
$$
where 
\begin{equation*}
C_j := \begin{cases}
D_0^{\top} \left( \prod_{i=1}^{i_1-1} D_i \right) A_{i_1} & \text{ if } j= 1, \\
 B_{i_{j-1}}^{\top} \left(\prod_{i=i_{j-1}+1}^{i_{j}-1} D_i \right) A_{i_{j}}
& \text{ if } 2 \leq j \leq m, \\
B_{i_m}^{\top} \left( \prod_{i=i_m+1}^{d} D_i \right) D_{d+1} &  \text{ if } j = m+1. 
\end{cases}
\end{equation*}
Clearly, for all $j \in [m+1]$, $C_j$ can be computed by a sparse-invertible-factor ROABP.
\end{proof}

Now, from the above lemma it is easy to construct a hitting-set. 
First we write a general result about hitting-sets
for a product of polynomials from some class \cite[Observation 4.1]{SY10}.
\begin{lemma}[Lagrange interpolation]
\label{lem:lagrange}
Suppose $\Hit$ is a hitting-set for a class of polynomials $\C$.
Let $C(\mathbf{x}) = C_1(\mathbf{x}) C_2(\mathbf{x}) \dotsm C_m(\mathbf{x})$,
where $C_i \in \C$ and has degree bounded by $\delta$, for all $i \in [m]$.
There is a hitting-set of size $m \delta \abs{\Hit} +1$ for $C(\mathbf{x})$.
\end{lemma}
\begin{proof}
Let $h = \abs{\Hit}$  and 
$\Hit = \{ \mathbf{\alpha}_1, \mathbf{\alpha}_2, \dots, \mathbf{\alpha}_{h} \}$.
Let $B := \{\beta_i\}_{i=1}^{h}$ be a set of constants.
The Lagrange interpolation $\mathbf{\alpha}(u)$ of the points in $\Hit$ is defined
as follows 
$$ \mathbf{\alpha}(u) := \sum_{i=1}^{h} \frac{\prod_{j \neq i} (u - \beta_j) }{\prod_{j \neq i} (\beta_i - \beta_j) } \mathbf{\alpha}_i . 
$$
The key property of the interpolation is that 
when we put $u = \beta_i$,
$\mathbf{\alpha}(\beta_i) = \mathbf{\alpha}_i$ for all $i \in [h]$.
For any $a \in [m]$, we know that 
$C_a(\mathbf{\alpha}_i) \neq 0$, for some $i \in [h]$. 
Hence, $C_a(\mathbf{\alpha}(u))$ as a polynomial in $u$ is nonzero 
because $C_a(\mathbf{\alpha}(\beta_i)) = C_a(\mathbf{\alpha}_i) \neq 0$.
So, we can say $C(\mathbf{\alpha}(u)) \neq 0$ as a polynomial in $u$. 
Degree of $\mathbf{\alpha}(u)$ is $h$. So, degree of $C(\mathbf{\alpha}(u))$
in $u$ is bounded by $m \delta h$. 
We can put $(m \delta h +1)$-many distinct values of $u$ 
to get a hitting-set for $C(\mathbf{\alpha}(u))$.
\end{proof}

Note that a hitting-set for $\alpha(\mathbf{x}) C(\mathbf{x})$ is also 
a hitting-set for $C(\mathbf{x})$ if $\alpha$ is a nonzero polynomial.
Recall that we get a hitting-set for invertible ROABP from
 Theorem~\ref{thm:invROABPHS}. 
Lemma~\ref{lem:width2} tells us how to 
write a width-$2$ ROABP as a product of width-$2$ invertible ROABPs. 
Combining these results with Lemma~\ref{lem:lagrange}
we directly get the following. 

\begin{theorem}
\label{thm:ROABP22HS}
Let $C(\mathbf{x}) = D_0^{\top}(\mathbf{x}_0) (\prod_{i=1}^{d} D_i(\mathbf{x}_i)) D_{d+1}(\mathbf{x}_{d+1})$
be a polynomial 
in $\F[\mathbf{x}]$ computed by a width-$2$ ROABP
 such that for all $0 \leq i \leq d+1$,
 $D_i$ has degree bounded
by $\delta$, $\sp(D_i) \leq s$ and $\mu(D_i) \leq \mu$.
Let $\ell := 1 + 2 \min\{ \ceil{\log (4 \cdot s)}, \mu \}$.
Then there is a hitting-set of size
$\poly((n \delta s)^{\ell})$.
\end{theorem} 

We remark again that when 
all $D_i$s are constant-variate or linear polynomials,
the hitting-set is polynomial-time.

\section{Discussion}

The first open problem is to do basis isolation for ROABP with only a polynomially large
weight assignment.
Also, our technique of finding a basis isolating weight assignment seems general. 
It needs to be explored, for what other general classes can it be applied.
In particular, can it be used to solve depth-$3$ multilinear circuits?
An easier question, perhaps, could be to improve Theorem~\ref{thm:baseSetsHS}
to get a truly blackbox PIT for the $2$-base-sets-1-distance model.

Another question is whether we can find a similar result in the boolean setting,
i.e.\ get a psuedorandom generator for unknown order ROBP with seed length same as
the known order case.

In the case of constant width ROABP, 
we could show constant-support concentration, but
only after assuming that the factor matrices are invertible. 
It seems that the invertibility assumption restricts the computing power of ROABP
significantly.
It is desirable to have low-support concentration without the
assumption of invertibility. 

As in the case of invertible ROABP and width-$2$ ROABP, analogous results
hold in the boolean setting, it will be interesting to see if there is some connection,
at the level of techniques,
between pseudorandom generators for boolean and arithmetic models.   

\section{Acknowledgements}
We thank Chandan Saha for suggestions to improve this paper. 
Several useful ideas about $\di$-distance circuits and 
base sets came up during discussions with him.
We thank Michael Forbes for suggesting a possible reduction from $\di$-distance circuits
to ROABP (Lemma~\ref{lem:dDistROABP}).
We thank anonymous reviewers for the various simplifications and useful suggestions. 
RG thanks TCS research fellowship for support.
NS thanks DST-SERB for the funding support.

\bibliographystyle{amsalpha}
\bibliography{deltaDistance}

\end{document}